\documentclass[10pt,journal]{IEEEtran}

\usepackage{amsmath}
\usepackage{amsthm}
\usepackage{amsfonts}
\usepackage{epsfig}
\usepackage{caption}
\usepackage{subcaption}
\usepackage{MnSymbol}
\usepackage{graphicx}
\usepackage{xcolor}
\usepackage{color}
\usepackage{tikz}
\usepackage{cite}
\usepackage{mathtools}
\usepackage{epstopdf}
\DeclareGraphicsExtensions{.pdf}
\newtheorem{definition}{Definition}

\newtheorem{theorem}{Theorem}
\newtheorem{lemma}{Lemma}

\usepackage{tikz}

\usetikzlibrary{calc,fadings,decorations.pathreplacing}

\newcommand\pgfmathsinandcos[3]{%
  \pgfmathsetmacro#1{sin(#3)}%
  \pgfmathsetmacro#2{cos(#3)}%
}
\newcommand\LongitudePlane[3][current plane]{%
  \pgfmathsinandcos\sinEl\cosEl{#2} 
  \pgfmathsinandcos\sint\cost{#3} 
  \tikzset{#1/.style={cm={\cost,\sint*\sinEl,0,\cosEl,(0,0)}}}
}
\newcommand\LatitudePlane[3][current plane]{%
  \pgfmathsinandcos\sinEl\cosEl{#2} 
  \pgfmathsinandcos\sint\cost{#3} 
  \pgfmathsetmacro\yshift{\cosEl*\sint}
  \tikzset{#1/.style={cm={\cost,0,0,\cost*\sinEl,(0,\yshift)}}} %
}
\newcommand\DrawLongitudeCircle[2][1]{
  \LongitudePlane{\angEl}{#2}
  \tikzset{current plane/.prefix style={scale=#1}}
  \pgfmathsetmacro\angVis{atan(sin(#2)*cos(\angEl)/sin(\angEl))} %
  \draw[current plane] (\angVis:1) arc (\angVis:\angVis+180:1);
  \draw[current plane,dashed] (\angVis-180:1) arc (\angVis-180:\angVis:1);
}
\newcommand\DrawLatitudeCircle[2][1]{
  \LatitudePlane{\angEl}{#2}
  \tikzset{current plane/.prefix style={scale=#1}}
  \pgfmathsetmacro\sinVis{sin(#2)/cos(#2)*sin(\angEl)/cos(\angEl)}
  \pgfmathsetmacro\angVis{asin(min(1,max(\sinVis,-1)))}
  \draw[current plane] (\angVis:1) arc (\angVis:-\angVis-180:1);
  \draw[current plane,dashed] (180-\angVis:1) arc (180-\angVis:\angVis:1);
}

\tikzset{%
  >=latex, 
  inner sep=0pt,%
  outer sep=2pt,%
  mark coordinate/.style={inner sep=0pt,outer sep=0pt,minimum size=3pt,
    fill=black,circle}%
}

\title{\LARGE \bf
{On Almost-Global Tracking for a Certain Class of Simple Mechanical Systems
}}
\author{A. Nayak$^{1}$  and R. N. Banavar$^{2}$
\thanks{$^{1}$ A. Nayak and $^{2}$R. N. Banavar are with the Systems and Control Enginerring,
Indian Institute of Technology Bombay, Mumbai, Maharashtra 400076, India
        {\tt\small aradhana@sc.iitb.ac.in, banavar@iitb.ac.in}}%
}

\begin{document}

\maketitle
\thispagestyle{empty}
\pagestyle{empty}

\begin{abstract}

In this article, we propose a control law for almost-global asymptotic tracking (AGAT) of a smooth reference trajectory for a fully actuated simple mechanical system (SMS) evolving
on a Riemannian manifold which can be embedded in a Euclidean space. The existing results on tracking for an SMS are either local, or almost-global, only in the case the manifold is a Lie group. In the latter case, the notion of a configuration error is naturally defined by the group operation and facilitates a global analysis. However, such a notion is not intrinsic to a Riemannian manifold. In this paper, we define a configuration error followed by
error dynamics on a Riemannian manifold, and then prove AGAT. The results are demonstrated for a spherical pendulum which is an SMS on $S^2$ and for a particle moving on a Lissajous curve in $\mathbb{R}^3$.

 \end{abstract}

 \section{Introduction}
The problem of stabilization of an equilibrium point of an SMS on a Riemannian manifold has been well studied in the  literature  in a geometric framework \cite{bulo}, \cite{3dpendulum}, \cite{cl1}, \cite{cl2}, \cite{ejc}. Further extensions of these results to the problem of locally tracking a smooth and bounded trajectory can be found in \cite{bulo}. An SMS is completely specified by a manifold, the kinetic energy of the mechanical system, which defines the Riemannian metric on the manifold, the potential forces, and the external forces or one forms on the manifold. If the Riemannian manifold is embedded in a Euclidean space, the metric on the manifold is induced from the Euclidean metric. In \cite{bulo}, a proportional and derivative plus feed forward (PD+FF) feedback control law is proposed for tracking a trajectory on a Riemannian manifold using error functions. This controller achieves asymptotic tracking only when the initial configuration of the SMS is in a neighborhood of the initial reference configuration. Therefore, such a tracking law achieves local convergence. As pointed out in \cite{kodi} and \cite{dayawansa}, global stabilization and global tracking is guaranteed only when the configuration manifold is diffeomorphic to $\mathbb{R}^n$. This leads us to the question: Is almost-global asymptotic stabilization (AGAS) of an equilibrium point and, almost-global asymptotic tracking (AGAT) of a suitable class of reference trajectories possible on a Riemannian manifold?
\newline
AGAS problems on a compact Riemannian manifold trace their origin to an early work by Koditschek. In \cite{kodi}, a potential function called as ``navigation function" is introduced, which is a Morse function on the manifold with a unique minimum. It is shown that there exists a dense set from which the trajectories of a negative gradient vector field generated by the navigation function converge to the minimum. A class of simple mechanical systems can be generated by the ``lifting" of the gradient vector field to the tangent bundle of the manifold. It is shown that the integral curves of such an SMS on the tangent bundle behave similar to integral curves of the gradient vector field of the navigation function on the manifold. In particular, the integral curve of such an SMS originating from a dense set in the tangent bundle converges asymptotically to the minimum of the navigation function in the zero section of the tangent bundle.
\newline
Koditschek's idea can be extended to almost-global tracking of a smooth and bounded trajectory on a Lie group. In a tracking problem, it is essential to define the notion of both configuration and velocity errors between the reference and the system trajectory. For an SMS on a lie group, a configuration error is defined by the left or right group operation and the velocity error is defined on the Lie algebra. This defines the error dynamics on the tangent bundle of the Lie group. It is shown in \cite{dayawansa} that the error dynamics is an SMS on the tangent bundle of the Lie group and, is generated by the tangent lift of a navigation function. Therefore, Koditschek's theorem can be applied to achieve AGAS of the error dynamics. Specific problems of almost-global tracking and stabilization in Lie groups have been studied in the literature as well. In \cite{asanyal}, \cite{tleemleok}, control laws for AGAT are proposed on $SE(3)$ using Morse functions.
\newline
Our contribution extends the existing results on AGAT to compact manifolds embedded in an Euclidean space. It is shown in \cite{Morse} that a navigation function exists on any compact manifold. We choose a configuration error map on the manifold subject to certain requirements imposed by the navigation function. The velocity error between the reference and system trajectory is defined along the error trajectory on the manifold with the help of two ``transport maps'' defined by the configuration error map. This construction gives rise to error dynamics on the tangent space of the error trajectory. To the best of our knowledge, this approach of AGAT for an SMS has not been explored before. In \cite{bulo}, a transport map is introduced to compare velocities at two configurations. However, in such a construction the error dynamics is not the ``lift" of a gradient vector field of a navigation function. Therefore, Koditschek's theorem is not applicable to the error dynamics.  In this paper, the error dynamics we introduce is an SMS on a compact Riemannian manifold and hence Koditschek's theorem can be applied for AGAS of the error dynamics. This leads to AGAT of the reference trajectory.
\newline
The paper is organised as follows. The second section is a brief introduction to relevant terminology in associated literature. In the third section we elaborate on the ``lift" of a gradient vector field and state the main result on AGAS from \cite{kodi}. The following section is on AGAT for a fully actuated SMS on a compact manifold. We first define the allowable configuration error map and navigation function for the tracking problem. Subsequently in the main theorem, we state our proposed control for AGAT on a Riemannian manifold. In the next section we demonstrate the idea of two transport maps for AGAT on a Lie group by choosing a configuration error defined by the group operation. The last section shows simulation results for a spherical pendulum which is an SMS on $S^2$ and for a particle moving on a Lissajous curve in $\mathbb{R}^3$.

\section{Preliminaries}
A Riemannian manifold is denoted by the 2-tuple $(M, \mathbb{G})$, where $M$ is a smooth connected manifold and $\mathbb{G}$ is a smooth, symmetric, positive definite $(0,2)$ tensor field (or a metric) on $M$. $\stackrel{\mathbb{G}}{\nabla}$ denotes the Riemannian connection on $(M, \mathbb{G})$ (see \cite{petersen},\cite{yano} for more details).  Let $\Psi: M \to \mathbb{R}$ be a twice differentiable function on $(M, \mathbb{G})$.
The Hessian of $\Psi$ is the symmetric $(0,2)$ tensor field denoted by $Hess \Psi$ and defined as $ Hess \Psi(q)(v_q, w_q) =  \langle \langle v_q, \stackrel{\mathbb{G}}{\nabla}_{w_q} grad \Psi \rangle \rangle,$ where $v_q$, $w_q \in T_q M$ and $\langle \langle v_q, w_q\rangle \rangle \coloneq \mathbb{G}(q)(v_q, w_q)$. Let $x_0$ be a critical point of $\Psi$ and $\{ x^1, \dotsc , x^n\}$ are local coordinates at $x_0$. The Hessian at $x_0$ is given in coordinates as ${(Hess \Psi(x_0))}_{ij} = \frac{\partial ^2 \Psi}{\partial x^i \partial x^j}(x_0)$ (see chapter 13 in \cite{milnor} for details). The map $\mathbb{G}^{\flat}(q) :  T_q M \to T_q ^* M$ is defined by $\langle \mathbb{G}(q)^{\flat} (v_1), v_2 \rangle \coloneq \mathbb{G}(q)(v_1, v_2)$ for $v_1$,$v_2 \in T_q M$. Therefore, if $\{ e^i\}$ is a basis for $T^*_q M$ in a coordinate system then $\mathbb{G}^{\flat}$ is expressed in coordinates as $ \mathbb{G}^{\flat}(q) (v_q) = \mathbb{G}_{ij} {v_q}^j e^i$, where $\mathbb{G}_{ij}$ is the matrix representation of $\mathbb{G}(q)$ in the chosen basis. The map $ \mathbb{G}^\sharp: T_q ^* M \to T_q M$ is dual of $\mathbb{G}^\flat$. It is expressed in coordinates as $ \mathbb{G}^\sharp (w) = \mathbb{G}^{ij} w_j e_i$ for $w \in T_q^*M$.
\subsection{SMS on a Riemannian manifold}
A fully actuated simple mechanical system (or an SMS) on a smooth, connected Riemannian manifold $(M, \mathbb{G})$ is denoted by the 3-tuple $(M,\mathbb{G}, F)$, where $F$ is an external force. The governing equations are
\begin{equation}\label{eq:3}
\stackrel{\mathbb{G}}{\nabla}_{\dot{\gamma}(t)} \dot{\gamma}(t) =  \mathbb{G}^{\sharp} (F(\dot{\gamma}(t)))
\end{equation}
where $\gamma(t)$ is the system trajectory.
\subsection{SMS on a Riemannian manifold embedded in $\mathbb{R}^m$}
Consider a Riemannian manifold $(M,\mathbb{G})$. By Nash embedding theorem (see \cite{nash}), there exists an isometric embedding $f : M \to \mathbb{R}^m$ for some $m$ depending on the dimension of $M$. The Euclidean metric $\mathbb{G}_{id}$ on $\mathbb{R}^m$ is the Riemannian metric such that in Cartesian coordinates
\begin{equation}\label{gid}
G_{id} = \delta_{ij} = \begin{cases}
1 & \text{if} \quad i =j\\
0 & \text{if} \quad i \neq j.
\end{cases}
\end{equation}
Therefore, the metric $\mathbb{G}$ on $M$ is induced by the Riemannian metric as follows
\begin{equation} \label{metricind}
\mathbb{G}= f^{*}\mathbb{G}_{id}
\end{equation}
where $f^{*}\mathbb{G}_{id}$ is the pull back of $\mathbb{G}_{id}$  (see Definition 3.81 in \cite{bulo}). The equations of motion for the SMS $(M, \mathbb{G},F)$ in \eqref{eq:3} can be simplified by embedding $M$ in $\mathbb{R}^m$. The idea behind this approach is that we consider the SMS to evolve on $(\mathbb{R}^m, \mathbb{G}_{id})$ subject to a distribution $\mathcal{D}$ (the velocity constraint) whose integral manifold is $M$ (see section 4.5 in \cite{bulo} for more details). The subspace $\mathcal{D}_{x} = T_{f^{-1}(x)} M$. Let $P_{\mathcal{D}_y}$ and $P_{\mathcal{D}_y}^{\perp}$ be projection bundle maps from
 $T_y \mathbb{R}^m$ to $\mathcal{D}_y$ and $\mathcal{D}^\perp_y$ respectively so that for $v_y \in T_y \mathbb{R}^m$, $P_{\mathcal{D}_y}(v_y) \in \mathcal{D}_y$ and $P_{\mathcal{D}_y}^\perp(v_y) \in \mathcal{D}^\perp_y$ respectively.
 \begin{definition}\label{constcon}(\cite{bulo})
 The constrained affine connection on $M$ is denoted by $\stackrel{\mathcal{D}}{\nabla}$ and is defined for $X$, $Y \in \Gamma^{\infty} (TM)$ as
\begin{equation}\label{eq:111}
  \stackrel{\mathcal{D}}{\nabla}_X Y = \stackrel{\mathbb{G}_{id}}{\nabla}_X Y + (\stackrel{\mathbb{G}_{id}}{\nabla}_X P_{\mathcal{D}}^\perp) Y
\end{equation}
 \end{definition}
The equations of motion for the SMS $(\mathbb{R}, \mathbb{G}_{id}, \mathcal{D})$ are (see Theorem 4.87 in \cite{bulo})
\begin{align}\label{123}
  \stackrel{\mathcal{D}}{\nabla} _{\dot{\boldsymbol{\gamma}}} \dot{\boldsymbol{\gamma}} &=0,\\
P_{\mathcal{D}_{\boldsymbol{\gamma}}}^\perp (\dot{\boldsymbol{\gamma}}(t_0))&=0 \quad \text{for some} \quad t_0 \in \mathbb{R}^{+},\nonumber
\end{align}
where $\boldsymbol{\gamma}:\mathbb{R}^{+} \to \mathbb{R}^m$ denotes the system trajectory, $\dot{\boldsymbol{\gamma}}(t) \in \mathcal{D}_{\boldsymbol{\gamma}(t)}$ for all $t \in \mathbb{R}^{+}$ and $\ddot{\boldsymbol{\gamma}}(t) \in {\mathcal{D}}^\perp_{\gamma(t)}$.   \eqref{123} can be simplified by substituting for $\stackrel{\mathcal{D}}{\nabla}$ from \eqref{eq:111} as follows
\begin{equation}\label{112}
 \ddot{\boldsymbol{\gamma}} +(\stackrel{\mathbb{G}_{id}}{\nabla}_{\dot{\boldsymbol{\gamma}}} P_{\mathcal{D}}^\perp)\dot{\boldsymbol{\gamma}}=0.
\end{equation}
Now consider the SMS $(\mathbb{R}, \mathbb{G}_{id},F, \mathcal{D})$. The equations are
\begin{align}\label{113}
\stackrel{\mathcal{D}}{\nabla} _{\dot{\boldsymbol{\gamma}}}\dot{\boldsymbol{\gamma}}&=  u,\\
P_{\mathcal{D}_{\boldsymbol{\gamma}}}^\perp (\dot{\boldsymbol{\gamma}}(t_0))&=0 \quad \text{for some} \quad t_0 \in \mathbb{R}^{+} \nonumber
\end{align}
where $\boldsymbol{\gamma}:\mathbb{R} \to \mathbb{R}^m$ denotes the system trajectory and, $u = P_{\mathcal{D}_\gamma} \{\mathbb{G}_{id}^\sharp (F)\}$. Further, substituting from \eqref{eq:111}, \eqref{113} can be written as
\begin{equation}\label{134}
 \ddot{\boldsymbol{\gamma}} +(\stackrel{\mathbb{G}_{id}}{\nabla}_{\dot{\boldsymbol{\gamma}}} P_{\mathcal{D}}^\perp)\dot{\boldsymbol{\gamma}}= u.
\end{equation}
\textit{Remark 1:} For a fully actuated SMS on $M$, the number of independent controls available is the dimension of $M$. However, as \eqref{113} and \eqref{134} are expressed in Euclidean coordinates, the control vector field $u$ is in $\mathbb{R}^m$.\\
\textit{Remark 2:} Equation \eqref{eq:3} and \eqref{113} both represent the dynamics of the SMS $(M, \mathbb{G},F)$. The system trajectory $\boldsymbol{\gamma}(t) \in \mathbb{R}^m$ (in \eqref{eq:3}) is the push forward of $\gamma(t) \in M$ (in \eqref{eq:111}) by the embedding $f: M \to \mathbb{R}^m$. Therefore, $\boldsymbol{\gamma}(t)= f_* \gamma(t)$ (see Remark 4.98 in \cite{bulo}).

\section{AGAS of error dynamics}
The main objective is almost global asympototic tracking (AGAT) of a given reference trajectory on the Riemannian manifold $(M, \mathbb{G})$ which is embedded in $\mathbb{R}^m$ (described in section II B.). The tracking problem is reduced to a stabilization problem by introducing the notion of a configuration error on the manifold. Almost global asymptotic stabilization (AGAS) of the error dynamics about a desired "zero error" configuration leads to AGAT of the reference trajectory. In this section, we (a) introduce this configuration error map between two configurations on a Riemmanian manifold to express the error between the reference and system trajectories and, (b) explicitly obtain the error dynamics for the tracking problem.
\\
\begin{definition}\label{navfn}(\cite{kodi})
A function $\psi: M \to \mathbb{R}$ on $(M, \mathbb{G})$ is a navigation function iff
\begin{enumerate}
  \item $\psi$ attains a unique minimum.
  \item $Det(Hess \psi(q)) \neq 0$ whenever $\mathrm{d}\psi(q) = 0$ for $q \in M$.
\end{enumerate}
\end{definition}
Let $\gamma:\mathbb{R}^+ \to M$ denote the system trajectory of $(M, \mathbb{G})$ and $\gamma_{ref}:\mathbb{R}^+ \to M$ denote a reference trajectory on the $M$. We define a $\mathcal{C}^2$ map $E:M \times M \to M$ between any two configurations on the manifold called the \textit{configuration error map}. The error trajectory on $M$ is $E(\gamma(t), \gamma_{ref}(t))$. The following condition characterizes the class of error maps for the AGAT on $M$.
\begin{definition}
  Consider a navigation function $\psi$ on $M$ (as in Definition ~\ref{navfn}). A configuration error map $E$ is compatible with $\psi$ iff
  \begin{itemize}
    \item $\psi\circ E(\gamma(t),\gamma_{ref}(t))= \psi \circ E(\gamma_{ref}(t),\gamma(t))$ for all $t \in \mathbb{R}^+$ and,
    \item $E(q,q) = q_0$, where $q_0$ is the minimum of $\psi$.
  \end{itemize}
\end{definition}
The following equation describes the controlled error dynamics for the tracking problem on $(M, \mathbb{G})$
\begin{align}\label{errdyn}
  &\stackrel{\mathcal{D}}{\nabla}_{\dot{\textbf{E}}(\boldsymbol{\gamma}(t), \boldsymbol{\gamma}_{ref}(t))}{\dot{\textbf{E}}(\boldsymbol{\gamma}(t), \boldsymbol{\gamma}_{ref}(t))} =\\
&P_{\mathcal{D}_{\textbf{E}(t)}}\{\mathbb{G}_{id}^\sharp (- K_p \mathrm{d}\boldsymbol{\psi}(\textbf{E})+F_{diss}(\dot{\textbf{E}} ))\} \nonumber
\end{align}
where $\boldsymbol{\gamma} = f_* \gamma$, $\boldsymbol{\gamma}_{ref} = f_* \gamma_{ref}$, $\boldsymbol{\psi} \coloneq f_* \psi$, $\textbf{E} \coloneq f_* E$, $E$ is a compatible with the navigation function $\psi$, $K_p$ is a positive definite matrix and, $F_{diss}: T\mathbb{R}^m \to T^* \mathbb{R}^m$ is a dissipative force, which means $\langle F_{diss}(v), v \rangle \leq 0$ for all $v \in \mathbb{R}^m$.
From \eqref{113}, we conclude that the error dynamics (in \eqref{errdyn}) is an SMS. In the following Lemma, we apply the main results from \cite{kodi} to conclude AGAS of the error dynamics about the minimum of $\psi$ lifted to the zero section of $TM$.\\
\textit{Note:} From henceforth we shall denote the  push forward of entities by $f$ in bold font.
\begin{lemma}\label{lem1}
 Consider the SMS $(M, \mathbb{G})$ whose dynamics is given by \eqref{113}, and a smooth reference trajectory $\gamma_{ref}: \mathbb{R}^+ \to M$. The error dynamics in \eqref{errdyn} is AGAS about $(q_m,0)$ where $q_m$ is the unique minimum of the navigation function $\psi$.
\end{lemma}

\begin{proof}
 We first rewrite \eqref{errdyn} so that the flow $\dot{E}(t)$ evolves on $TM$. From the equivalence in representation of Riemmannian connection $ \stackrel{\mathbb{G}}{\nabla}$ and the constrained connection $ \stackrel{\mathcal{D}}{\nabla}$ noted in Remark 2, \eqref{errdyn} can be expressed as
\begin{align}\label{eq:75}
 \dot{E} &= v_e \\ \nonumber
 \stackrel{\mathbb{G}}{\nabla}_{\dot{E}}{v_e} &= - \mathbb{G}^\sharp K_p \mathrm{d}\psi(E) + \mathbb{G}^\sharp F_{diss}(v_e) \nonumber
 \end{align}
We define an energy like function $E_{cl}$ on $TM$ as $E_{cl}(E,v_E) \coloneq K_p \psi(E) + \frac{1}{2} {||v_e ||}^2$. Then, $E_{cl}(q_m,0)=0$ and $E_{cl}(q,0)>0$ for all $(q,0)$ in a neighborhood of $(q_m,0)$. Also,
\begin{align*}
\frac{\mathrm{d}}{\mathrm{d}t} E_{cl} &= \langle K_p \mathrm{d} \psi(E),v_e  \rangle + \ll v_e, \stackrel{\mathbb{G}}{\nabla}_{\dot{E}} v_e \gg \\
&=\langle K_p \mathrm{d} \psi(E),v_e  \rangle + \mathbb{G}(v_{e},  -\mathbb{G}^{\sharp} (K_p \mathrm{d} \psi(E) - F_{diss}(v_{e}) )\\
&= K_p \langle \mathrm{d} \psi, v_{e} \rangle- K_p \langle \mathrm{d}\psi, v_{e} \rangle + \langle F_{diss}(v_{e}), v_{e} \rangle \leq 0
\end{align*}
as $F_{diss}$ is dissipative. Therefore $E_{cl}$ is a Lyapunov function and the error dynamics in \eqref{errdyn} is locally stable around $(q_m,0)$. In what follows, we obtain all parts of from proposition 3.6 in \cite{kodi} for the error dynamics.
\newline
In proposition 3.2 in \cite{kodi}, as $M$ is a manifold without boundary, $b_1 = +\infty$ and $b_0 = \psi(q_m)$. Therefore, $\epsilon^b = TM$ and by Propostion 3.2, $TM$ is a positively invariant set.
\newline
Observe that $(q*,0)$ is an equilibrium state of \eqref{errdyn} iff $q*$ is a critical point of $\psi$. By proposition 3.3 in \cite{kodi}, the positive limit set of all solution trajectories of \eqref{errdyn} originating in the positive invariant set $TM$ is $\{(p,0)\in TM: \mathrm{d}\psi(p)=0\}$.
\newline
In order to study the behavior of the flow of the vector field in \eqref{eq:75}, we linearize the equations around $(q*,0)$ to obtain
\begin{align}\label{linerr}
&\begin{pmatrix}
\dot{E}\\
\dot{v}_e
\end{pmatrix}=\\ \nonumber
&\begin{pmatrix}
0 & I_n\\
-G^{ij}(q*) {Hess \psi}(q*) & G^{ij}(q*) \circ \frac{\partial F_{diss}}{\partial v_e}(q*,0)
\end{pmatrix}\begin{pmatrix}
E\\ v_e\\
\end{pmatrix}.
\end{align}
By Lemma 3.5 in \cite{kodi} the origin of the LIT system \eqref{linerr} in $\mathbb{R}^{2n}$, $n$ being the dimension of $M$, is (a) asymptotically stable, (b) stable but not attractive or (c) unstable if the origin of the following LTI system in $\mathbb{R}^n$ has the corresponding property
\[\dot{E}= -Hess \psi(E)
\]
The local behavior of the error dynamics around $(q*,0)$ is therefore, determined by the nature of Hessian at $q^*$.
\newline
The trajectories of $-\mathrm{d}\psi(E(t))$ converge to $q_m$ from all but the stable manifolds of the maxima and saddle points. As $\psi$ is a navigation function, the stable manifolds of the maxima and saddle points constitutes a nowhere dense set. Therefore there is a dense set in $TM$ for which all trajectories of \eqref{errdyn} converge to $(q_m,0)$.
\end{proof}
\section{AGAT for an SMS}
In this section we propose a control law for AGAT of a reference trajectory for a fully actuated SMS $(M, \mathbb{G})$ for which the equations of motion are given in \eqref{113}. We first obtain a simplified expression for the constrained covariant derivative of two vector fields on $M$.
\begin{lemma} (Proposition 4.85 in \cite{bulo})\label{lem2}
  Let $X \in \Gamma^{\infty}(TM)$ and $Y \in \Gamma^{\infty}(\mathcal{D})$ be vector fields on $M$. The constrained covariant derivative of $Y$ along $X$ is given as
  \begin{equation}\label{116}
   \stackrel{\mathcal{D}}{\nabla}_X Y = P_{\mathcal{D}}(\stackrel{\mathbb{G}_{id}}{\nabla}_X Y ).
  \end{equation}
\end{lemma}
\begin{theorem}(AGAT)\label{thm1}
Consider the SMS $(M,\mathbb{G})$ given by \eqref{113} and a smooth trajectory $\gamma_{ref}:\mathbb{R} \to M$ with bounded velocity. Let $\psi: M \to \mathbb{R}$ be a navigation function and $E:M \times M \to M$ be a compatible error map on the manifold. Then there exists an open dense set $S \in TM$ such that AGAT of $\gamma_{ref}$ is achieved for all $(\gamma(0),\dot{\gamma}(0))\in S$ with $u$ in \eqref{113} given by the solution to the following equations

  \begin{align}\label{thm11}
   \mathrm{d}_1 \mathbf{E}(\boldsymbol{\gamma}, \boldsymbol{\gamma}_{ref})(u) &=  P_{\mathcal{D}_{\mathbf{E}}}(- K_p \mathrm{d}\boldsymbol{\psi}(\mathbf{E}) + F_{diss}(\dot{\mathbf{E}})\\ \nonumber
    &- \mathrm{d}_1 (\mathrm{d}_1 \mathbf{E} )(\dot{\boldsymbol{\gamma}},\dot{\boldsymbol{\gamma}}) \\ \nonumber
    &+ \mathrm{d}_2 (\mathrm{d}_1 \mathbf{E}) (\dot{\boldsymbol{\gamma}}_{ref},\dot{\boldsymbol{\gamma}})+ \frac{\mathrm{d}}{\mathrm{d}t}(\mathrm{d}_2 \mathbf{E} \dot{\boldsymbol{\gamma}}_{ref})) ,\\ \nonumber
\text{and}, \quad  P_{\mathcal{D}_{\boldsymbol{\gamma}}}^\perp (u)&=0.\nonumber
 \end{align}
where $F_{diss}: T\mathbb{R}^m \to T^*\mathbb{R}^m $ is a dissipative force and $K_p \in \mathbb{R}^+$.
\end{theorem}

\begin{proof}
  Let $\textbf{E}(\boldsymbol{\gamma}(t), \boldsymbol{\gamma}_{ref}(t))$ be error trajectory and the closed loop error dynamics be given by \eqref{errdyn}. The velocity vector of the error trajectory is given by
\begin{equation}\label{eq:73}
  \dot{\textbf{E}}(\boldsymbol{\gamma}(t), \boldsymbol{\gamma}_{ref}(t)) = \mathrm{d}_1\textbf{E}(\boldsymbol{\gamma}, \boldsymbol{\gamma}_{ref} ) \dot{\boldsymbol{\gamma}} + \mathrm{d}_2\textbf{E} (\boldsymbol{\gamma}, \boldsymbol{\gamma}_{ref}) \dot{\boldsymbol{\gamma}}_{ref}
\end{equation}
where $\mathrm{d}_1\textbf{E} (\boldsymbol{\gamma}, \boldsymbol{\gamma}_{ref}): T_{\boldsymbol{\gamma}} M \to T_{\textbf{E}(\boldsymbol{\gamma}, \boldsymbol{\gamma}_{ref})}M$ is the partial derivative of $\textbf{E}$ with respect to the first argument and, $\mathrm{d}_2\textbf{E} (\boldsymbol{\gamma}, \boldsymbol{\gamma}_{ref}): T_{{\boldsymbol{\gamma}}_{ref}} M \to T_{\textbf{E}(\boldsymbol{\gamma}, \boldsymbol{\gamma}_{ref})}M$ is the partial derivative of $\textbf{E}$ with respect to the second argument and, $\dot{\textbf{E}}(\boldsymbol{\gamma}, \boldsymbol{\gamma}_{ref}) \in T_{\textbf{E}(\boldsymbol{\gamma}, \boldsymbol{\gamma}_{ref})}M $.\\
 $\mathrm{d}_1\textbf{E}(\boldsymbol{\gamma}, \boldsymbol{\gamma}_{ref}) $ and  $\mathrm{d}_2\textbf{E}(\boldsymbol{\gamma}, \boldsymbol{\gamma}_{ref})$ are similar to ``transport maps" in \cite{bulo} as they transport vectors along the system and reference trajectory respectively to vectors along the error trajectory.
 \begin{figure}[h]
  \centering
  \includegraphics[scale=0.5]{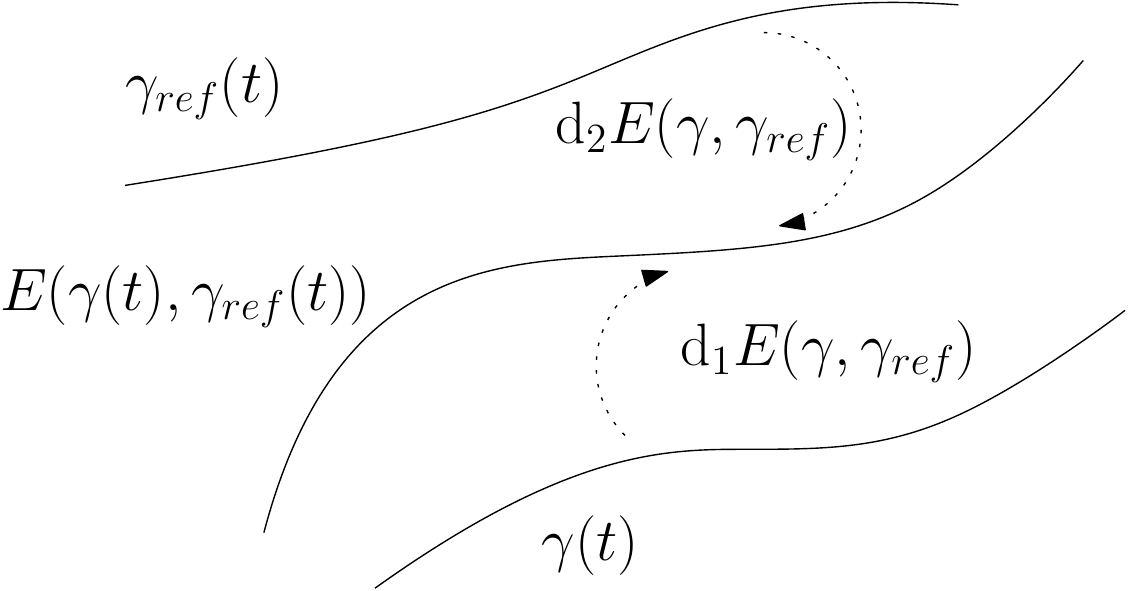}
  \caption{Two transport maps from controlled trajectory $\boldsymbol{\gamma}(t)$ and reference trajectory $\boldsymbol{\gamma}_{ref}(t)$ to the error trajectory $\textbf{E}(\boldsymbol{\gamma},\boldsymbol{\gamma}_{ref})$}\label{2transportmaps}
\end{figure}
\begin{align}\label{118}
\stackrel{\mathcal{D}}{\nabla}_{\dot{\textbf{E}}}{\dot{\textbf{E}}}&= \stackrel{\mathcal{D}}{\nabla}_{\dot{\textbf{E}}}(\mathrm{d}_1\textbf{E} (\boldsymbol{\gamma}, \boldsymbol{\gamma}_{ref})\dot{\boldsymbol{\gamma}} + \mathrm{d}_2\textbf{E}(\boldsymbol{\gamma}, \boldsymbol{\gamma}_{ref}) \dot{\boldsymbol{\gamma}}_{ref})
\end{align}
As $\mathrm{d}_1\textbf{E} (\boldsymbol{\gamma}, \boldsymbol{\gamma}_{ref})\dot{\boldsymbol{\gamma}}$ and $\mathrm{d}_1\textbf{E} (\boldsymbol{\gamma}, \boldsymbol{\gamma}_{ref})\dot{\boldsymbol{\gamma}}_{ref}$ are vector fields along $E
(t)$ on $Q$, therefore, by Lemma~\ref{lem2},
\begin{subequations}\label{117}
  \begin{equation}
    \stackrel{\mathcal{D}}{\nabla}_{\dot{\textbf{E}}}(\mathrm{d}_1\textbf{E} (\boldsymbol{\gamma}, \boldsymbol{\gamma}_{ref})\dot{\boldsymbol{\gamma}}) = P_{\mathcal{D}_\textbf{E}}(\stackrel{\mathbb{G}_{id}}{\nabla}_ {\dot{\textbf{E}}}\mathrm{d}_1\textbf{E} (\boldsymbol{\gamma}, \boldsymbol{\gamma}_{ref})\dot{\boldsymbol{\gamma}} )
  \end{equation}
  and,
  \begin{equation}
     \stackrel{\mathcal{D}}{\nabla}_{\dot{\textbf{E}}}(\mathrm{d}_2\textbf{E} (\boldsymbol{\gamma}, \boldsymbol{\gamma}_{ref})\dot{\boldsymbol{\gamma}}_{ref}) = P_{\mathcal{D}_\textbf{E}}(\stackrel{\mathbb{G}_{id}}{\nabla}_ {\dot{\textbf{E}}}\mathrm{d}_2\textbf{E} (\boldsymbol{\gamma}, \boldsymbol{\gamma}_{ref})\dot{\boldsymbol{\gamma}}_{ref})
  \end{equation}
\end{subequations}
\textit{Note: We shall drop arguments and refer to $\textbf{E}(\boldsymbol{\gamma}, \boldsymbol{\gamma}_{ref})$ as $\textbf{E}$ and similarly refer to $\mathrm{d}_1\textbf{E} (\boldsymbol{\gamma}, \boldsymbol{\gamma}_{ref}) $ and $\mathrm{d}_2\textbf{E} (\boldsymbol{\gamma}, \boldsymbol{\gamma}_{ref})$ as $\mathrm{d}_1\textbf{E}$, $\mathrm{d}_2\textbf{E}$ respectively.}
\\
From \eqref{118} and \eqref{117},
\begin{align}\label{119}
  \stackrel{\mathcal{D}}{\nabla}_{\dot{\textbf{E}}}{\dot{\textbf{E}}} &= P_{\mathcal{D}_\textbf{E}}(\stackrel{\mathbb{G}_{id}}{\nabla}_ {\dot{\textbf{E}}}\mathrm{d}_1\textbf{E} \dot{\boldsymbol{\gamma}} + \stackrel{\mathbb{G}_{id}}{\nabla}_ {\dot{\textbf{E}}}\mathrm{d}_2\textbf{E} \dot{\boldsymbol{\gamma}}_{ref}) \\ \nonumber
   &=   P_{\mathcal{D}_\textbf{E}}(\frac{\mathrm{d}}{\mathrm{d}t}(\mathrm{d}_1\textbf{E} \dot{\boldsymbol{\gamma}}) + \stackrel{\mathbb{G}_{id}}{\nabla}_ {\dot{\textbf{E}}}\mathrm{d}_2\textbf{E} \dot{\boldsymbol{\gamma}}_{ref})
\end{align}
$\mathrm{d}_1 \textbf{E}(\boldsymbol{\gamma}, \boldsymbol{\gamma}_{ref})$ is a $(1,1)$ tensor which depends on two configurations at which error is defined. Therefore $\mathrm{d}_1 (\mathrm{d}_1 \textbf{E}): TM \times TM \to TM$ and $\mathrm{d}_2 (\mathrm{d}_1 \textbf{E}):TM \times TM \to TM$ are $(2,1)$ tensors. The first term in the bracket can be expressed in terms of these $(2,1)$ tensors as follows.
\begin{align}\label{120}
 \frac{\mathrm{d}}{\mathrm{d}t} (\mathrm{d}_1\textbf{E} \dot{\boldsymbol{\gamma}}) &= (\mathrm{d}_1 (\mathrm{d}_1 \textbf{E}))(\dot{\boldsymbol{\gamma}}, \dot{\boldsymbol{\gamma}}) + (\mathrm{d}_2 (\mathrm{d}_1 \textbf{E}))(\dot{\boldsymbol{\gamma}}_{ref}, \dot{\boldsymbol{\gamma}})\\ \nonumber
 &+ \mathrm{d}_1\textbf{E}(\stackrel{\mathcal{D}}{\nabla}_{\dot{\boldsymbol{\gamma}}} \dot{\boldsymbol{\gamma}})
\end{align}
In the last term in \eqref{120} we consider the constrained covariant derivative to differentiate $\dot{\boldsymbol{\gamma}}$ as $\mathrm{d}_1\textbf{E}$ is a transport map from $T_{\boldsymbol{\gamma}}M $ to $T_\textbf{E} M$. From equations \eqref{119} and \eqref{120},
\begin{align}\label{121}
  \stackrel{\mathcal{D}}{\nabla}_{\dot{\textbf{E}}}{\dot{\textbf{E}}} &= P_{\mathcal{D}_\textbf{E}}(
  (\mathrm{d}_1 (\mathrm{d}_1 \textbf{E}))(\dot{\boldsymbol{\gamma}}, \dot{\boldsymbol{\gamma}}) + (\mathrm{d}_2 (\mathrm{d}_1 \textbf{E}))(\dot{\boldsymbol{\gamma}}_{ref}, \dot{\boldsymbol{\gamma}})\\ \nonumber
  &+ \stackrel{\mathbb{G}_{id}}{\nabla}_ {\dot{\textbf{E}}}\mathrm{d}_2\textbf{E} \dot{\boldsymbol{\gamma}}_{ref}+ \mathrm{d}_1\textbf{E}(\stackrel{\mathcal{D}}{\nabla}_{\dot{\boldsymbol{\gamma}}}\dot{\boldsymbol{\gamma}})).
\end{align}
From \eqref{113}, \eqref{121}, and using $P_{\mathcal{D}_{\textbf{E}}}(\mathrm{d}_1\textbf{E} (u))=\mathrm{d}_1\textbf{E} (u)$,
\begin{align}\label{122}
 \mathrm{d}_1\textbf{E} (u)&=  \stackrel{\mathcal{D}}{\nabla}_{\dot{\textbf{E}}}\dot{\textbf{E}} - P_{\mathcal{D}(\textbf{E})}(
  (\mathrm{d}_1 (\mathrm{d}_1 \textbf{E}))(\dot{\boldsymbol{\gamma}}, \dot{\boldsymbol{\gamma}}) \\ \nonumber
  &+ (\mathrm{d}_2 (\mathrm{d}_1 \textbf{E}))(\dot{\boldsymbol{\gamma}}_{ref}, \dot{\boldsymbol{\gamma}})
  + \stackrel{\mathbb{G}_{id}}{\nabla}_ {\dot{\textbf{E}}} \mathrm{d}_2\textbf{E} \dot{\boldsymbol{\gamma}}_{ref}).
\end{align}
Substituting for $\stackrel{\mathcal{D}}{\nabla}_{\dot{\textbf{E}}}\dot{\textbf{E}}$ from \eqref{errdyn} in Lemma~\ref{lem1} we get \eqref{thm11}. As the error dynamics is AGAS, therefore, $u$ in \eqref{122} leads to AGAT of $\gamma_{ref}(t)$.
\end{proof}
\textit{Remark 1:} \eqref{thm11} is the solution to an underdetermined system of $m-$equations. Recall that $u \in \mathbb{R}^m$ is the representation of available independent controls in the $m-$ dimensional Euclidean space. \\
\textit{Remark 2:} The Nash embedding theorem ensures that a Riemannian manifold can be embedded isometrically in some Euclidean space. Therefore, Theorem~\ref{thm1} reduces the problem of almost-global tracking of a given reference trajectory on a Riemannian manifold to finding a navigation function $\psi$ and a compatible configuration error map $E$ on the manifold. It is well known that a navigation function exists on a compact manifold (\cite{kodi}, \cite{milnor}). The compatible error map is obtained from the embedding as will be seen in examples.\\
\textit{Remark 3:} The control law given in \cite{buloandmurray} uses a single transport map $\mathcal{T}(\gamma, \gamma_{ref})$ to transport $\dot{\gamma}_{ref}$ to $T_{\gamma} M$. Instead of this, we have two transport maps $\mathrm{d}_1 E $ and $\mathrm{d}_2 E$ to transport $\dot{\gamma}$ and $\dot{\gamma}_{ref}$ respectively to $T_{E(\gamma, \gamma_{ref})} M$ along the error trajectory $\dot{E}$. The tracking error function $\Psi: M \times M \to \mathbb{R}$ in \cite{bulo} is similar to $\psi(E)$. However, instead of the velocity error along $\dot{\gamma}(t)$, we consider the velocity error $\dot{E}$ along $E(t)$. This is the essential difference in our approach to tracking a trajectory for an SMS.\\
\textit{Remark 4:} In this remark we follow the procedure in \cite{buloandmurray} to obtain error dynamics and show that the theorem in \cite{kodi} cannot be applied to conclude AGAT even if a navigation function is chosen as a potential function. Let us consider the tracking error function $\Psi: M \times M \to \mathbb{R}$ defined as $\Psi = \psi \circ E$ for a navigation function $\psi$ and a compatible error map $E$. The control law for local tracking of $\gamma_{ref}$ in \cite{buloandmurray} is
\begin{align*}
F(\gamma, \dot{\gamma}) &= -d_1 \Psi(\gamma, \gamma_{ref}) - F_{diss}(\dot{\gamma} -\mathcal{T}(\gamma, \gamma_{ref}) \dot{\gamma}_{ref})\\
&+{\mathbb{G}}^\flat(\stackrel{\mathbb{G}}{\nabla}_{\dot{\gamma}}  \mathcal{T}(\gamma, \gamma_{ref}) \dot{\gamma}_{ref}(t)+ \frac{\mathrm{d}}{\mathrm{d}t} (\mathcal{T} \dot{\gamma}_{ref}) )
\end{align*}
where $\mathcal{T}(\gamma, \gamma_{ref}):T_{\gamma_{ref}}M \to T_{\gamma}M $ is a transport map compatible with the error function $\Psi$ (\cite{buloandmurray}). The velocity error is defined along $\gamma(t)$ as $v'_e \coloneqq \dot{\gamma} - \mathcal{T} \dot{\gamma}_{ref}$. The error dynamics in this case is given by
\begin{subequations}
      \begin{equation}\label{eq:23}
    \dot{E} = \mathrm{d}_1E . v'_{e}
    \end{equation}
    \begin{equation} \label{eq:54}
     \dot{v'}_e = \mathbb{G}^{\sharp} (-\mathrm{d}\psi (E) \mathrm{d}_1 E + F_{diss}(v'_{e})) - I(v'_e,\mathcal{T}\dot{\gamma}_{ref} ) -C(v'_{e})
    \end{equation}
    \end{subequations}
where $I(v'_e,\mathcal{T}\dot{\gamma}_{ref} ) =\Gamma_{ij}^k {v'_e}^i{(\mathcal{T}\dot{\gamma}_{ref})}^j$ and $C(v'_e) = \Gamma_{ij}^k {v'_e}^i {v'_e}^j$. \eqref{eq:23} results from the following equivalent compatibility condition
      \begin{equation}\label{eq:93}
 \mathrm{d}_2 E(\gamma, \gamma_{ref}) = - \mathrm{d}_1 E(\gamma, \gamma_{ref}) \mathcal{T}(\gamma, \gamma_{ref})
\end{equation}
and \eqref{eq:54} is given by the following identity
\begin{align}\label{eq:5}
\stackrel{\mathbb{G}}{\nabla}_{\dot{\gamma}(t)} v'_{e} &= \stackrel{\mathbb{G}}{\nabla}_{\dot{\gamma}} (\dot{\gamma}- \mathcal{T}(\gamma, \gamma_{ref}). \dot{\gamma}_{ref}(t))\\ \nonumber
&= \mathbb{G}^{\sharp} (F_{PD} + F_{FF}) - \mathbb{G}^{\sharp}F_{FF}\\ \nonumber
&= -\mathbb{G}^{\sharp} (d_1 \Psi (\gamma, \gamma_{ref})) +\mathbb{G}^{\sharp} (F_{diss}(v'_{e})). \nonumber
\end{align}

Linearizing \eqref{eq:23}-\eqref{eq:54} about an equilibrium state $(E*,0)$,
      \begin{align*}
        &\begin{pmatrix}
        \dot{E}\\
        \dot{v}'_e
        \end{pmatrix}= \\
        &\begin{pmatrix}
        0 & \mathrm{d}_1E\\
        -G^{ij} {\mathrm{d}^2 \psi}(E^*)\mathrm{d}_1E & G^{ij} \circ F_{diss}(E^*,0) - I^{\flat} (\mathcal{T} \dot{\gamma}_{ref})
        \end{pmatrix} \begin{pmatrix}
        E\\ v'_e\\
        \end{pmatrix}.
        \end{align*}
As $I^{\flat} (\mathcal{T} \dot{\gamma}_{ref})$ is a time dependent term, the flow of error dynamics around $(E^*,0)$ cannot be determined by the flow of $- \mathrm{d}\psi$. As a result, the lifting property of dissipative systems cannot be used to establish AGAS of the error dynamics for the PD+FF tracking control law in \cite{buloandmurray}.\\
\section{AGAT for an SMS on a Lie group}
In this section, we utilize the idea of two transport maps originating from the configuration error map to study AGAT for an SMS on a compact Lie group. The configuration error map is defined using the group operation. Therefore, the problem of AGAT for an SMS on a compact Lie group is reduced to choosing a compatible navigation function.
\subsection{Preliminaries}
Let $G$ be a Lie group and let $\mathfrak{g}$ denote its Lie algebra. Let $\phi: G \times G \to G$ be the left group action in the first argument defined as $\phi(g,h) \coloneq  L_{g} (h)= gh $ for all $g$, $h \in G$. The Lie bracket on $\mathfrak{g}$ is $[,]$. The adjoint map $ad_\xi : \mathfrak{g} \to \mathfrak{g}$ for $\xi \in \mathfrak{g}$ and defined as $ad_\xi \eta \coloneq [\xi,\eta]$. Let $\mathbb{I} :\mathfrak{g} \to \mathfrak{g}^*$ be an isomorphism on the Lie algebra to its dual and the inverse is denoted by $\mathbb{I}^\sharp: \mathfrak{g}^* \to \mathfrak{g}$. $\mathbb{I}$ induces a
left invariant metric on $G$ (see section 5.3 in \cite{bulo}). This metric on $G$ is denoted by $\mathbb{G}_{\mathbb{I}}$ and defined as $\mathbb{G}_{\mathbb{I}}(g).(X_g,Y_g) \coloneq \langle \mathbb{I}(T_gL_{g^{-1}} (X_g)),T_gL_{g^{-1}} (Y_g)\rangle$ for all $g \in G$ and $X_g$, $Y_g \in T_g G$. The equations of motion for the SMS $(G, {\mathbb{I}},F)$ where $F \in \mathfrak{g}^*$ are given by
\begin{align}\label{dynliegrp}
\xi &= T_g L_{g^{-1}} \dot{g},\\ \nonumber
\dot{\xi} - I^\sharp ad^*_\xi I \xi &= \mathbb{I}^{\sharp} F
\end{align}
where $g(t)$ describes the system trajectory. $\xi(t)$ is called the body velocity of $g(t)$.

\begin{lemma}\label{lem3}
Given a differentiable parameterized curve $\gamma: \mathbb{R} \to G$ and a vector field $X$ along $\gamma(t)$ we have the following equality
\[ \stackrel{\mathbb{G}_{\mathbb{I}}}{\nabla}_{\dot{\gamma}} X = T_eL_{\gamma}( \frac{\mathrm{d}}{\mathrm{d}t}
(T_{\gamma}L_ {\gamma^{-1}} X(t) ) +\stackrel{\mathfrak{g}}{\nabla}_{\xi(t)} T_{\gamma}L_ {\gamma^{-1}}X(t))
\]
where $\stackrel{\mathfrak{g}}{\nabla}$ is the bilinear map defined as
\begin{equation}\label{eq:60}
 \stackrel{\mathfrak{g}}{\nabla}_\eta \nu = \frac{1}{2} [\eta, \nu] - \frac{1}{2} \mathbb{I}^ \sharp ( ad^*_{\eta} \mathbb{I}\nu +  ad^*_{\nu} \mathbb{I}\eta  )
\end{equation}
for $\nu$, $\eta \in \mathfrak{g}$.
\end{lemma}

\begin{proof}
Let $\{ e_1, \dotsc , e_n\}$ be a basis for $\mathfrak{g}$. Let,
\begin{align*} X(\gamma(t)) &= T_eL_{\gamma} \big(\sum_{i=1}^{n} v_{X(\gamma(t))}^i(t) e_i   \big) \\
&= \sum_{i=1}^{n} v_{X(\gamma(t))}^i(t) (e_i)_L (\gamma(t))
\end{align*}
where $v_X= T_{\gamma}L_ {\gamma^{-1}}X(\gamma(t))$ and $(e_i)_L(g) = T_eL_g e_i $ and therefore $(e_i)_L(g)$
is a basis for $T_gG$ for all $g \in G$.
Similarly let
\[ \dot{\gamma}(t) = T_eL_{\gamma} \big(\sum_{j=1}^{n} v_{\gamma}^j(t) e_j   \big)
= \big(\sum_{j=1}^{n} v_{\gamma}^j(t) (e_j)_L (\gamma(t))  \big)
\]
where $v_{\gamma}(t)= T_{\gamma}L_ {\gamma^{-1}}\dot{\gamma}(t)$.
Using properties of affine connection we have,
\begin{align*}
   \stackrel{\mathbb{G}_{\mathbb{I}}}{\nabla}_{\dot{\gamma}} X
   &=  \stackrel{\mathbb{G}_{\mathbb{I}}}{\nabla}_{ v_{\gamma}^j(t) (e_j)_L (\gamma(t))} v_{X}^i(t) (e_i)_L (\gamma(t)) \\
  &=  \frac{\mathrm{d}}{\mathrm{d}t}(v_X^i)(e_i)_L(\gamma(t))
  + v_X^k v_\gamma^j (\stackrel{\mathfrak{g}}{\nabla}_{(e_j)_L}(e_k)_L ) (\gamma(t)) \\
  &= T_eL_\gamma \big( \frac{\mathrm{d}}{\mathrm{d}t}(v_X^i) \big)e_i +  v_X^k v_\gamma^j \stackrel{\mathfrak{g}}{\nabla}_{(e_j)}(e_k)  \\
  &= T_eL_{\gamma}( \frac{\mathrm{d}}{\mathrm{d}t} (T_{\gamma}L_ {\gamma^{-1}} X(t))
  +\stackrel{\mathfrak{g}}{\nabla}_{\xi(t)} T_{\gamma}L_ {\gamma^{-1}}X(t))
\end{align*}
\end{proof}
\begin{definition}
 The configuration error $E: G \times G \to G$ map is defined as
\begin{equation} \label{errmap}
E(g, h) = L_{h}g^{-1}.
\end{equation}
\end{definition}
\subsection{AGAT for an SMS on a compact Lie group}
\begin{theorem}\label{thm2}
  (AGAT for Lie group) Let $G$ be a compact Lie group and $\mathbb{I}: \mathfrak{g} \to \mathfrak{g}^*$ be an isomorphism on the Lie algebra. Consider the SMS on the Riemannian manifold $(G, {\mathbb{I}})$ given by \eqref{dynliegrp} and a smooth reference trajectory $g_r :\mathbb{R} \to G$ on the Lie group which has bounded velocity. Let $\psi : G \to \mathbb{R}$ be a navigation function compatible with the error map in \eqref{errmap}. Then there exists an open dense set $S \in G \times \mathfrak{g}$ such that AGAT of $g_r$ is achieved for all $(g(0), \xi(0)) \in S$ with $u = \mathbb{I}^\sharp(F)$ in \eqref{dynliegrp} given by the following equation
  \begin{align}\label{thm2eqn}
    u&= - g^{-1}_r \mathbb{G}_{\mathbb{I}}^\sharp (-K_p \mathrm{d}\psi (E) + F_{diss}\dot{E}) g + g^{-1} (\stackrel{\mathfrak{g}}{\nabla}_\eta \eta  \\ \nonumber
    &+  \frac{\mathrm{d}}{\mathrm{d}t}{E^{-1}}\mathrm{d}_2E(\dot{g}_r)) g - I^\sharp ad^*_\xi I \xi
  \end{align}
  where $\eta$ is the body velocity of the error trajectory and $K_p >0$.
\end{theorem}
\begin{figure}[h!]
  \centering
  \includegraphics[scale=0.35]{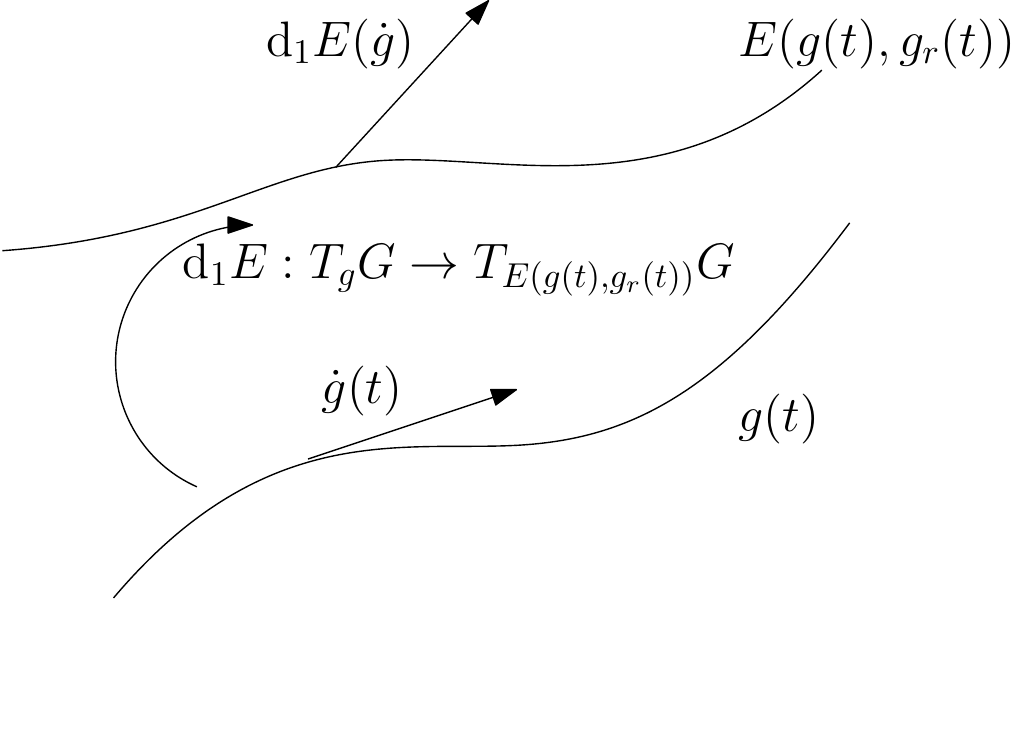}
  \caption{$\mathrm{d}_1E(\dot{g})$ is a vector field along $E(t)$}
\end{figure}

\begin{proof}
The error trajectory is $E(g(t), g_r(t))$ where $E$ is defined in \eqref{errmap}. The error dynamics on $(G, {\mathbb{I}})$ is similar to \eqref{errdyn} with the appropriate Riemannian connection as follows
\begin{equation}\label{errdynlie}
\stackrel{\mathbb{G}_{\mathbb{I}}}{\nabla}_{\dot{E}} \dot{E}  = \mathbb{G}_{\mathbb{I}}^\sharp (-K_p \mathrm{d}\psi (E) + F_{diss}\dot{E})
\end{equation}
As $(E, \psi)$ is a compatible pair, by Lemma ~\ref{lem1}, the error dynamics is AGAS about the minimum of $\psi$. The derivative of the error trajectory is given by \eqref{eq:73}. Therefore, $\mathrm{d}_1E(\dot{g})= -g_r g^{-1} \dot{g}g^{-1} = T_g L_{g_r g^{-1}}R_{g^{-1}}\dot{g}$ and $\mathrm{d}_2 E(\dot{g}_r)= \dot{g}_r g^{-1} = T_{g_r}R_{g^{-1}} \dot{g}_r$ and,
\begin{align}\label{eq:55}
 \stackrel{\mathbb{G}_{\mathbb{I}}}{\nabla}_{\dot{E}} \dot{E} &=
 \stackrel{\mathbb{G}_{\mathbb{I}}}{\nabla}_{\dot{E} }(\mathrm{d}_1E(\dot{g}) + \mathrm{d}_2E(\dot{g}_r) ).\\ \nonumber
\end{align}
As $\mathrm{d}_1E(\dot{g})$ and $\mathrm{d}_2E(\dot{g}_r)$ are vector fields along $E(t)$, we use Lemma ~\ref{lem3} to expand $ \nabla_{\dot{E} }(\mathrm{d}_1E(\dot{g}))$ as
\begin{align*}
 \stackrel{\mathbb{G}_{\mathbb{I}}}{\nabla}_{\dot{E}} \mathrm{d}_1E(\dot{g}) &= T_eL_{E}( \frac{\mathrm{d}}{\mathrm{d}t}
(T_{E}L_ {E^{-1}} \mathrm{d}_1E(\dot{g})) \\ &+\stackrel{\mathfrak{g}}{\nabla}_{T_{E}L_ {E^{-1}} \dot{E}} T_{E}L_ {E^{-1}}\mathrm{d}_1E(\dot{g}) )\\
&= E( \frac{\mathrm{d}}{\mathrm{d}t}
(E^{-1}\mathrm{d}_1E(\dot{g}))  +\stackrel{\mathfrak{g}}{\nabla}_ { E^{-1}\dot{E}} {E^{-1}}\mathrm{d}_1E(\dot{g}) )\\
&= E( g \dot{\xi} g^{-1} + \stackrel{\mathfrak{g}}{\nabla}_ { E^{-1}\dot{E}} {E^{-1}}\mathrm{d}_1E(\dot{g}))
\end{align*}
Similarly, the second term in \eqref{eq:55} is
\begin{align*}
\stackrel{\mathbb{G}_{\mathbb{I}}}{\nabla}_{\dot{E} }\mathrm{d}_2E(\dot{g}_r) &=  \frac{\mathrm{d}}{\mathrm{d}t}
(E^{-1}\mathrm{d}_2E(\dot{g}_r))+   \stackrel{\mathfrak{g}}{\nabla}_ { E^{-1}\dot{E}} {E^{-1}}\mathrm{d}_2E(\dot{g}_r)
\end{align*}
Hence from \eqref{eq:55},
\begin{align}\label{eq:51}
\stackrel{\mathbb{G}_{\mathbb{I}}}{\nabla}_{\dot{E}} \dot{E} &= E( g \dot{\xi} g^{-1} + \stackrel{\mathfrak{g}}{\nabla}_ { E^{-1}\dot{E}} {E^{-1}}\mathrm{d}_1E(\dot{g})\\ \nonumber
&+  \frac{\mathrm{d}}{\mathrm{d}t} (E^{-1}\mathrm{d}_2E(\dot{g}_r))+   \stackrel{\mathfrak{g}}{\nabla}_ { E^{-1}\dot{E}} {E^{-1}}\mathrm{d}_2E(\dot{g}_r))\\ \nonumber
&=  E( g \dot{\xi} g^{-1} + \stackrel{\mathfrak{g}}{\nabla}_ { E^{-1}\dot{E}} E^{-1}\dot{E} +  \frac{\mathrm{d}}{\mathrm{d}t}{E^{-1}}\mathrm{d}_2E(\dot{g}_r))  \nonumber
\end{align}
Let $u_1=\mathbb{G}_{\mathbb{I}}^\sharp (- \mathrm{d}\psi (E) + F_{diss}\dot{E})$. From \eqref{dynliegrp}, $ \dot{\xi} = u + I^\sharp ad^*_\xi I \xi$ and $\eta= E^{-1} \dot{E}$, therefore, \eqref{eq:55} is
\begin{align}
u_1 &= E(-g(u+I^\sharp ad^*_\xi I \xi ) g^{-1} +\stackrel{\mathfrak{g}}{\nabla}_\eta \eta  +  \frac{\mathrm{d}}{\mathrm{d}t}({E^{-1}}\mathrm{d}_2E(\dot{g}_r)))  \nonumber
\end{align}
which gives
\begin{align}\label{eq:53}
  u &=  g^{-1} (- E^{-1}u_1+ \stackrel{\mathfrak{g}}{\nabla}_\eta \eta  +  \frac{\mathrm{d}}{\mathrm{d}t}{E^{-1}}\mathrm{d}_2E(\dot{g}_r)) g - I^\sharp ad^*_\xi I \xi \\ \nonumber
  &= - g^{-1}_r u_1 g + g^{-1} (\stackrel{\mathfrak{g}}{\nabla}_\eta \eta  +  \frac{\mathrm{d}}{\mathrm{d}t}{E^{-1}}\mathrm{d}_2E(\dot{g}_r)) g - I^\sharp ad^*_\xi I \xi \nonumber
\end{align}
\end{proof}
\begin{figure}[!ht]
  \begin{subfigure}[b]{0.18\textwidth}
  \centering
  \includegraphics[scale=0.055]  {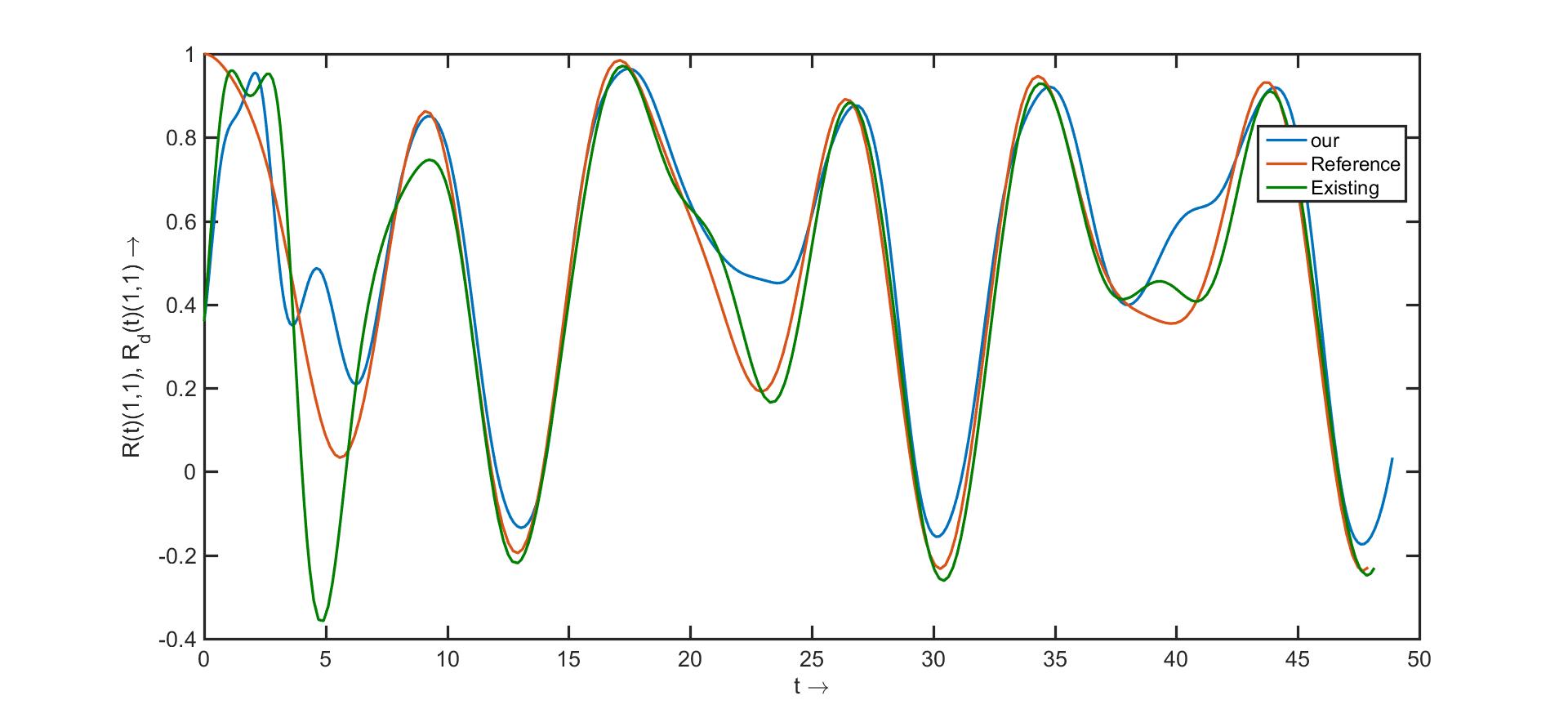}
  \caption{$(1,1)$(t)}\label{fig7}
  \end{subfigure}
\begin{subfigure}[b]{0.18\textwidth}
  \centering
  \includegraphics[scale=0.052]{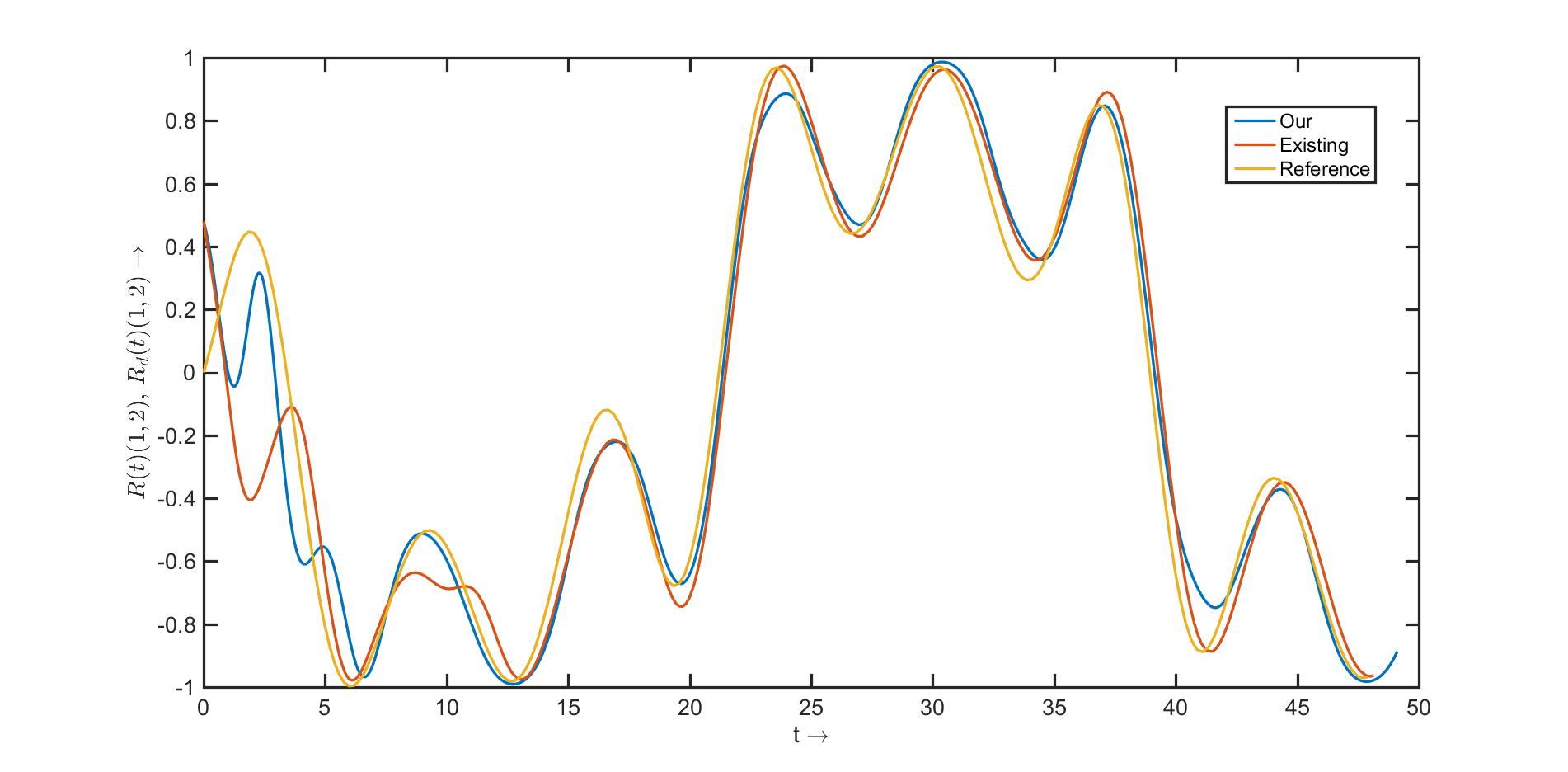}
  \caption{$(1,2)$(t)}\label{fig8}
\end{subfigure}
  \caption{A comparison of tracking results}
\end{figure}
\begin{figure}[ht]
  \begin{subfigure}[b]{0.18\textwidth}
  \centering
  \includegraphics[scale=0.06]  {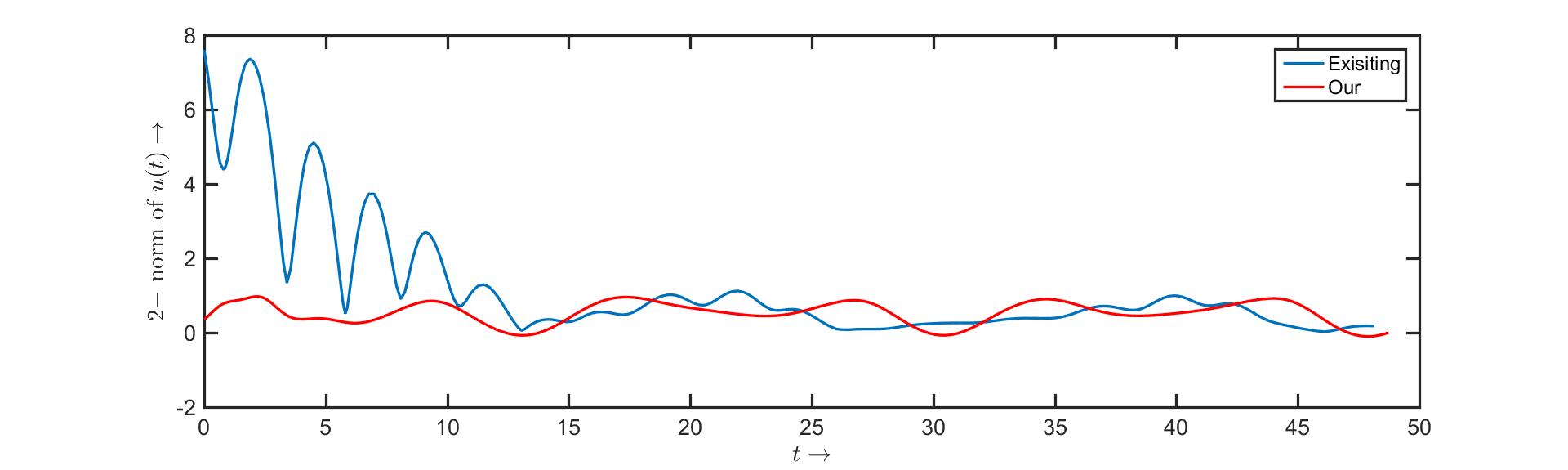}
  \caption{(i)}\label{fig9}
  \end{subfigure}
\begin{subfigure}[b]{0.18\textwidth}
  \centering
  \includegraphics[scale=0.06]{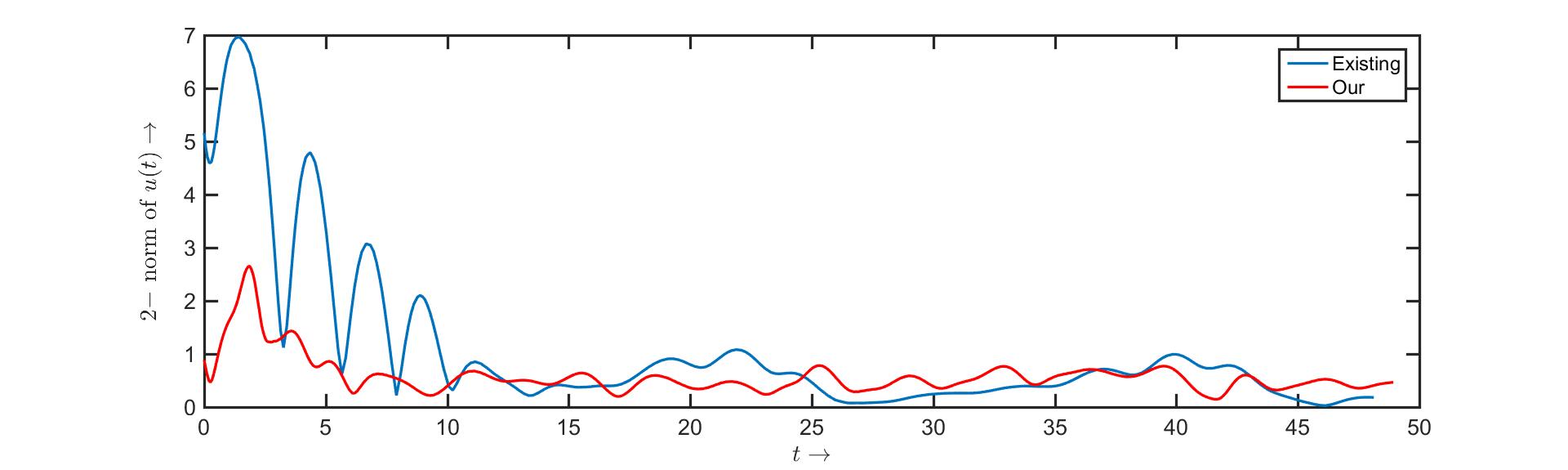}
  \caption{(ii)}\label{fig10}
\end{subfigure}
  \caption{Comparison of control effort for initial conditions (i) and (ii)}
\end{figure}
\textit{Remark 1:}
The control law in \cite{buloandmurray} for tracking the reference trajectory $t \to g_r(t) \in G$ by a fully actuated SMS given by $(G, \mathbb{I}, \mathbb{R}^n)$ is
\begin{align} \label{132}
&F(t,g,\xi) = - T_eL_{g^{-1}_{r} g}^* (\mathrm{d} \Psi (g^{-1}_{r} g)) +F_{diss} (\xi - Ad_{g^{-1} {g^{-1}_{r}}}v_r ) \\ \nonumber
&+ \mathbb{I}^\sharp (\stackrel{\mathfrak{g}}{\nabla}_{\xi} Ad_{g^{-1} {g_r}}v_r + [Ad_{g^{-1}g_r} v_r, \xi ] + Ad_{g^{-1}g_r} \dot{v}_{r} )
\end{align}
where $v_r(t)$ is the body velocity of the reference trajectory defined by $v_r(t) = T_{g_r} L_{g_r^{-1}} \dot{g}_{r}(t)$ and $\Psi : G \times G \to \mathbb{R}$ is a tracking error function. It is shown in \cite{dayawansa} that by choosing $\Psi = \psi \circ E$, where $E$ is defined as in \eqref{errmap} and $\psi$ is a navigation function, the control law in \eqref{132} achieves AGAT of $g_r (t)$. On comparing \eqref{eq:53} and \eqref{132} with $\Psi = \psi \circ E$, it is observed that in \eqref{eq:53} the acceleration of the error trajectory on the Lie algebra given by $\stackrel{\mathfrak{g}}{\nabla}_\eta \eta$ appears as an additional term.
\\
In order to observe the effect of this term on the controlled trajectory we compare the tracking results for an externally actuated rigid body obtained by the existing the control law with the proposed control law. The rigid body is an SMS on $SO(3)$ and $\psi(E)= tr(P(I_3-E))$, where $P$ is a symmetric positive definite matrix is chosen as the compatible navigation function which has a unique minimum at $I_3$. We consider a rigid body with an inertia matrix given by $\mathbb{I} =\begin{pmatrix} 4&1 &1\\ 1 5 & 0.2&  2\\ 1& 2 & 6.3\\ \end{pmatrix}$ and initial conditions $R(0) = \begin{pmatrix} 0.36&  0.48 &-0.8 \\ -0.8 &0.6& 0\\ 0.48 & 0.64 & 0.60\\ \end{pmatrix}$ and $\mathbb{I}\Omega = \begin{pmatrix} 1 &2.2 &5.1\\ \end{pmatrix}$. and the reference is generated by a dummy rigid body with inertia matrix $\mathbb{I}_d=\begin{pmatrix}
               1 & 0& 0\\ 0 &1.2& 0\\ 0& 0 &2 \\
               \end{pmatrix}$, initial conditions, $R_d(0)=I_3 $ and $\mathbb{I}_d{\Omega_d} = \begin{pmatrix}-0.8& -0.3 &-0.5\end{pmatrix}$.
In the Morse function, $P = diag(4 ,4.5 ,4.2)$ and in the intermediate control $F_{diss} = -diag( 3.5 ,3.5, 3.7)$. In figures ~\ref{fig7} and ~\ref{fig8}, the reference and two controlled trajectories obtained by the existing and proposed control law are plotted together.

In order to compare the control effort, we compute the $2$ norm of $\breve{u}(t) \in \mathfrak{so}(3)$ for (i) the rigid body with the above initial conditions and (ii) with initial conditions given by $R(0) = \begin{pmatrix} 0.7071&  0.7071 & 0 \\ -0.7071 & 0.7071 & 0\\ 0 & 0 & 1\\ \end{pmatrix}$ and $\mathbb{I}\Omega = \begin{pmatrix} 1 &2.2 &5.1\\ \end{pmatrix}$ and compare it with the $2$ norm of existing control in figures \ref{fig9} and ~\ref{fig10} respectively.

\section{Simulation Results}
\subsection{AGAT on $S^2$}
\begin{figure}[h!]
\begin{tikzpicture} 
\def\R{1} 
\def\angEl{35} 
\def\angAz{-105} 
\def\angPhi{-40} 
\def\angBeta{19} 


\pgfmathsetmacro\H{\R*cos(\angEl)} 
\tikzset{xyplane/.estyle={cm={cos(\angAz),sin(\angAz)*sin(\angEl),-sin(\angAz),
                              cos(\angAz)*sin(\angEl),(0,-\H)}}}
\LongitudePlane[xzplane]{\angEl}{\angAz}
\LongitudePlane[pzplane]{\angEl}{\angPhi}
\LatitudePlane[equator]{\angEl}{0}


\draw[xyplane] (-2*\R,-2*\R) rectangle (2.2*\R,2.8*\R);
\fill[ball color=white] (0,0) circle (\R); 
\draw (0,0) circle (\R);



\coordinate (O) at (0,0);
\coordinate [mark coordinate](N) at (0,\H);
\coordinate [mark coordinate](S) at (0,-\H);
\path[pzplane] (\R,0) coordinate (PE);
\path[xzplane] (\R,0) coordinate (XE);
\path (PE) ++(0,-\H) coordinate (Paux); 
                                        second line={(S)--(Paux)});


\DrawLatitudeCircle[\R]{0} 
\DrawLongitudeCircle[\R]{\angAz} 
\DrawLongitudeCircle[\R]{\angAz+90} 
\DrawLongitudeCircle[\R]{\angPhi} 


\draw[xyplane,<->] (1.8*\R,0) node[below] {$y$} -- (0,0) -- (0,2.4*\R)
    node[right] {$x$};
\draw[->] (0,-\H) -- (0,1.6*\R) node[above] {$z$};



\end{tikzpicture}
\caption{A 2-sphere in $\mathbb{R}^3$}
\end{figure}
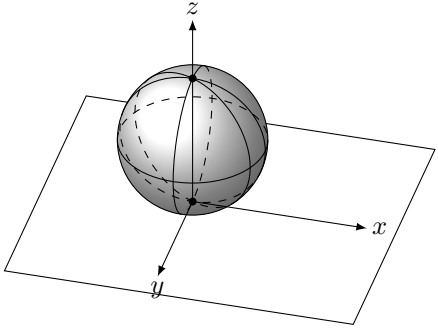
\subsubsection{Navigation function}
We consider the restriction of the height function in $\mathbb{R}^3$ to $S^2$ given by $\boldsymbol{\psi}(x,y,z)=z $ for $(x,y,z)\in \mathbb{R}^3$. It is a navigation function on $S^2$ with $(0,0,-1)^T$ as the unique minimum and $(0,0,1)^T$ as maximum. It can be verified that the $Hess \boldsymbol{\psi}$ is non-degenerate at both extremal points. The projection map $P_{\mathcal{D}_q}:\mathbb{R}^3 \to T_qS^2$ is defined as
\begin{equation}\label{125}
  P_{\mathcal{D}_q}v= - \{\hat{q}\}^2v
\end{equation}
\subsubsection{Configuration Error map}
The configuration error map $\textbf{E}: S^2 \times S^2 \to S^2$ is chosen as
\begin{equation}\label{126}
  \textbf{E}(q_1,q_2)= \begin{pmatrix}
                \sqrt{1- {\langle q_1,q_2\rangle_{\mathbb{R}^3}}^2} & 0 & {-\langle q_1,q_2\rangle_{\mathbb{R}^3}}^2
              \end{pmatrix}^T
\end{equation}
for $q_1$, $q_2 \in S^2$. As $E$ is symmetric, therefore, $\boldsymbol{\psi} \circ \textbf{E}$ is also symmetric. $\textbf{E}(q,q)= \begin{pmatrix}
                                                             0 & 0 & -1
                                                           \end{pmatrix}^T$ hence, $E(q,q)$ is the minimum of $\boldsymbol{\psi}$. Therefore $(\boldsymbol{\psi},\textbf{E})$ is a compatible pair according to Definition 1.\\
The $(1,1)$ tensors are
\begin{align*}
\mathrm{d}_1\textbf{E}(q_1,q_2)&= \begin{pmatrix}
                            \frac{\langle q_1,q_2\rangle_{\mathbb{R}^3}q_2^T}{ \sqrt{1- {\langle q_1,q_2\rangle_{\mathbb{R}^3}}}} & 0_{1 \times 3} & -q_2^T
                          \end{pmatrix}^T \quad \text{and,}\\
\mathrm{d}_2\textbf{E}(q_1,q_2)&= \begin{pmatrix}
                            \frac{\langle q_1,q_2\rangle_{\mathbb{R}^3}q_1^T}{ \sqrt{1- {\langle q_1,q_2\rangle_{\mathbb{R}^3}}^2}} & 0_{1 \times 3} & -q_1^T
                          \end{pmatrix}^T.
\end{align*}
The $(2,1)$ tensors are $3 \times 3$ arrays given by
\begin{align*}
  \mathrm{d}_1\mathrm{d}_1\textbf{E}(q_1,q_2)&= \begin{pmatrix}
                            \frac{q_2q_2^T}{ {\sqrt{1- {\langle q_1,q_2\rangle_{\mathbb{R}^3}}^2}^3}} & 0_{3\times 3} & 0_{3 \times 3}
                          \end{pmatrix}^T,\\
  \mathrm{d}_2\mathrm{d}_1\textbf{E}(q_1,q_2)&= \begin{pmatrix}
                            \frac{q_1q_2^T -{\langle q_1,q_2\rangle}_{\mathbb{R}^3}^3 I_3 + {\langle q_1,q_2\rangle}_{\mathbb{R}^3} I_3}{ {\sqrt{1- {\langle q_1,q_2\rangle_{\mathbb{R}^3}}^2}^3}} & 0_{3\times 3} & -I_3
                          \end{pmatrix}^T,\\
  \mathrm{d}_2\mathrm{d}_2\textbf{E}(q_1,q_2)&= \begin{pmatrix}
                            \frac{q_1q_1^T}{ {\sqrt{1- {\langle q_1,q_2\rangle_{\mathbb{R}^3}}^2}^3}} & 0_{3\times 3} & 0_{3 \times 3}
                          \end{pmatrix}^T \quad \text{and,}\\
\mathrm{d}_1\mathrm{d}_2\textbf{E}(q_1,q_2)&= \begin{pmatrix}
                            \frac{q_2q_1^T -{\langle q_1,q_2\rangle}_{\mathbb{R}^3}^3 I_3 + {\langle q_1,q_2\rangle}_{\mathbb{R}^3} I_3}{ {\sqrt{1- {\langle q_1,q_2\rangle_{\mathbb{R}^3}}^2}^3}} & 0_{3\times 3} & -I_3
                          \end{pmatrix}^T.
\end{align*}
\subsubsection{AGAT results}
\begin{figure}[h!]
  \centering
  \begin{subfigure}[b]{0.15\textwidth}
  \includegraphics[scale=0.1]  {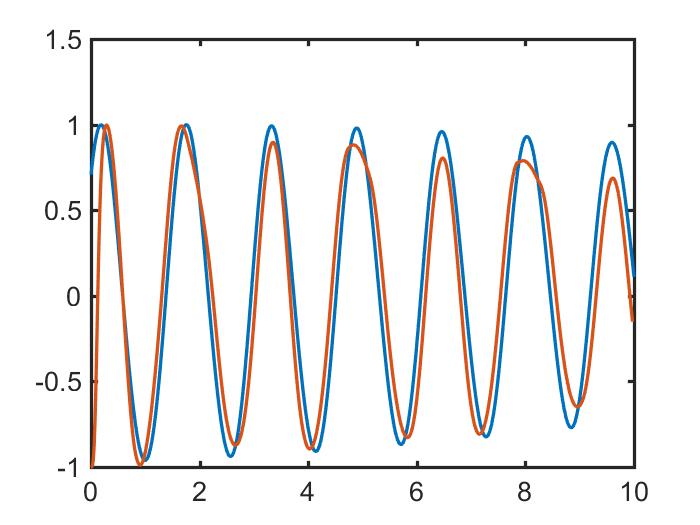}
  \caption{x coordinate}\label{fig1}
  \end{subfigure}
\begin{subfigure}[b]{0.15\textwidth}
  \centering
  \includegraphics[scale=0.1]{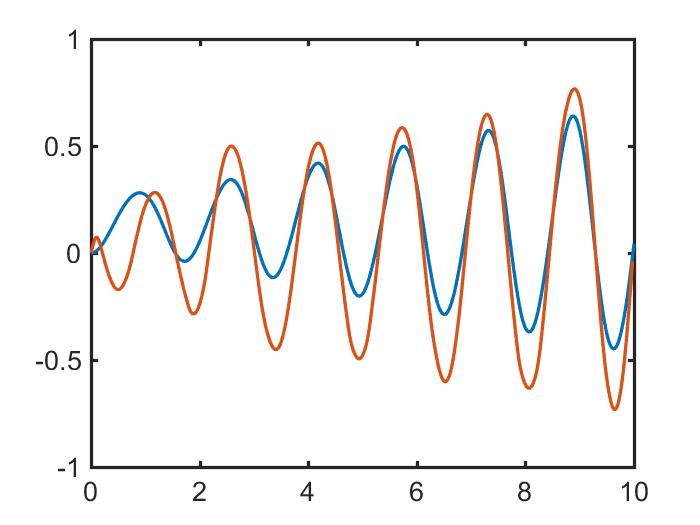}
  \caption{y coordinate}\label{fig2}
\end{subfigure}
%
\begin{subfigure}[b]{0.15\textwidth}
  \centering
  \includegraphics[scale=0.1]{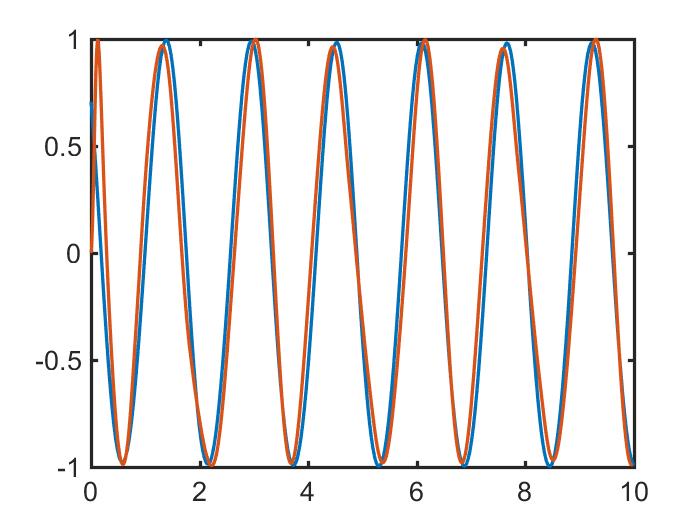}
  \caption{z coordinate}\label{fig3}
  \end{subfigure}
  \caption{Tracking results for first set of initial conditions}
\end{figure}
\begin{figure}[h!]
  \centering
  \begin{subfigure}[b]{0.15\textwidth}
  \includegraphics[scale=0.1]  {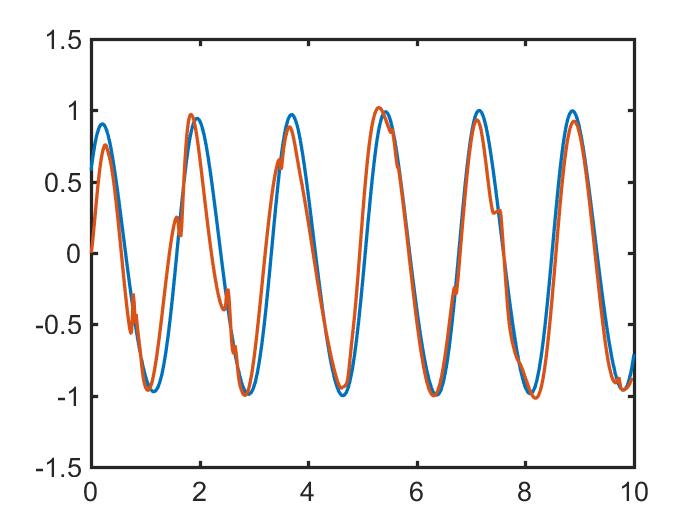}
  \caption{x coordinate}\label{fig4}
  \end{subfigure}
\begin{subfigure}[b]{0.15\textwidth}
  \centering
  \includegraphics[scale=0.1]{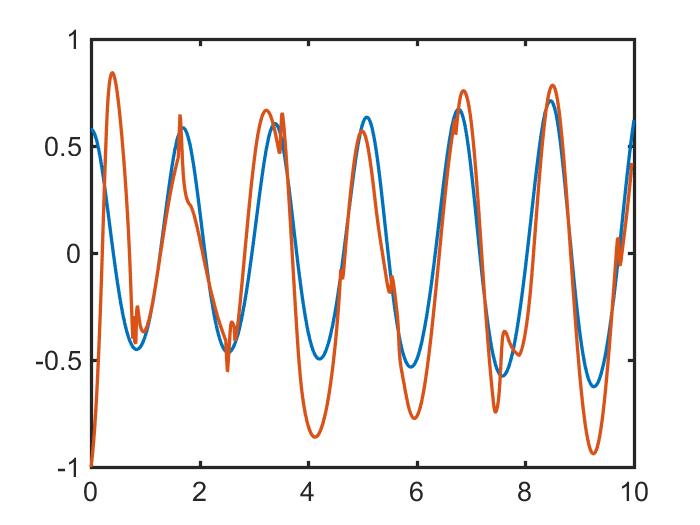}
  \caption{y coordinate}\label{fig5}
\end{subfigure}
\begin{subfigure}[b]{0.15\textwidth}
  \centering
  \includegraphics[scale=0.1]{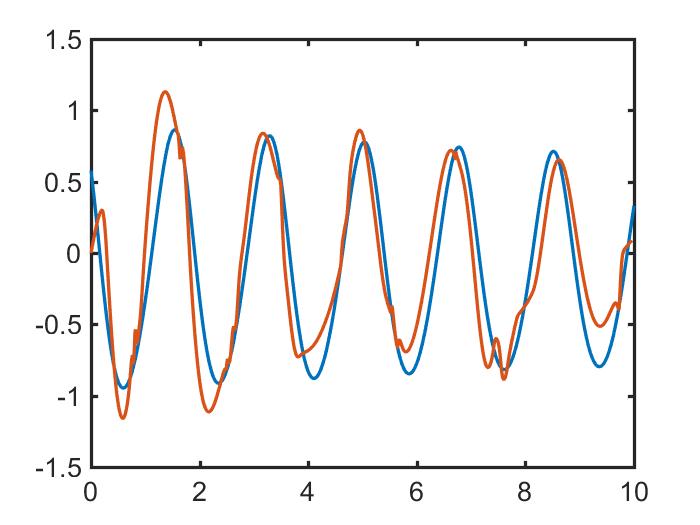}
  \caption{z coordinate}\label{fig6}
  \end{subfigure}
  \caption{Tracking results for second set of initial conditions}
\end{figure}
\begin{figure}[h!]
  \centering
  \begin{subfigure}[b]{0.15\textwidth}
  \includegraphics[scale=0.1]  {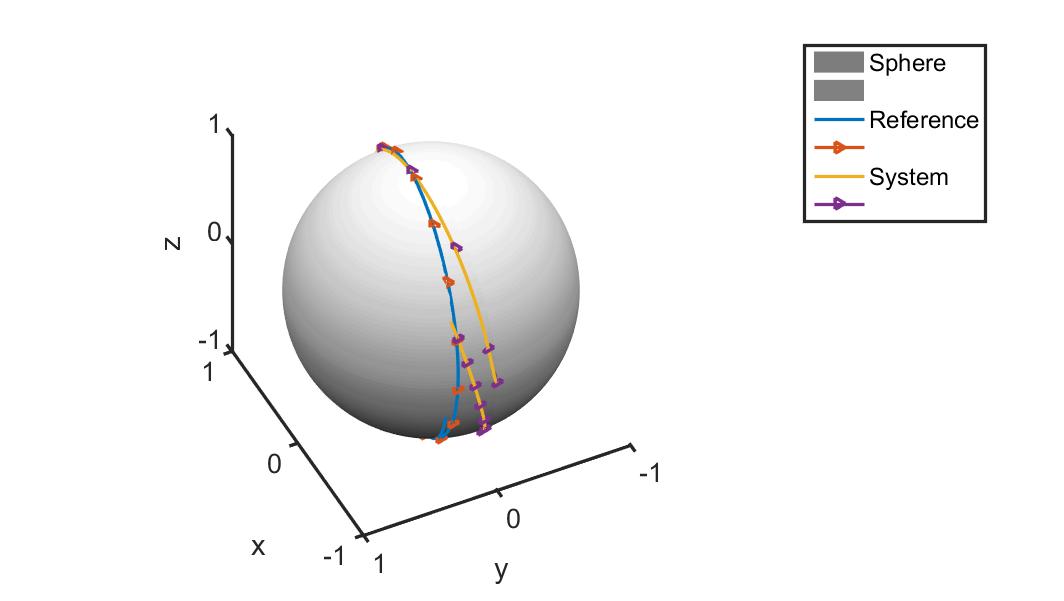}
  \end{subfigure}
\begin{subfigure}[b]{0.15\textwidth}
  \centering
  \includegraphics[scale=0.1]{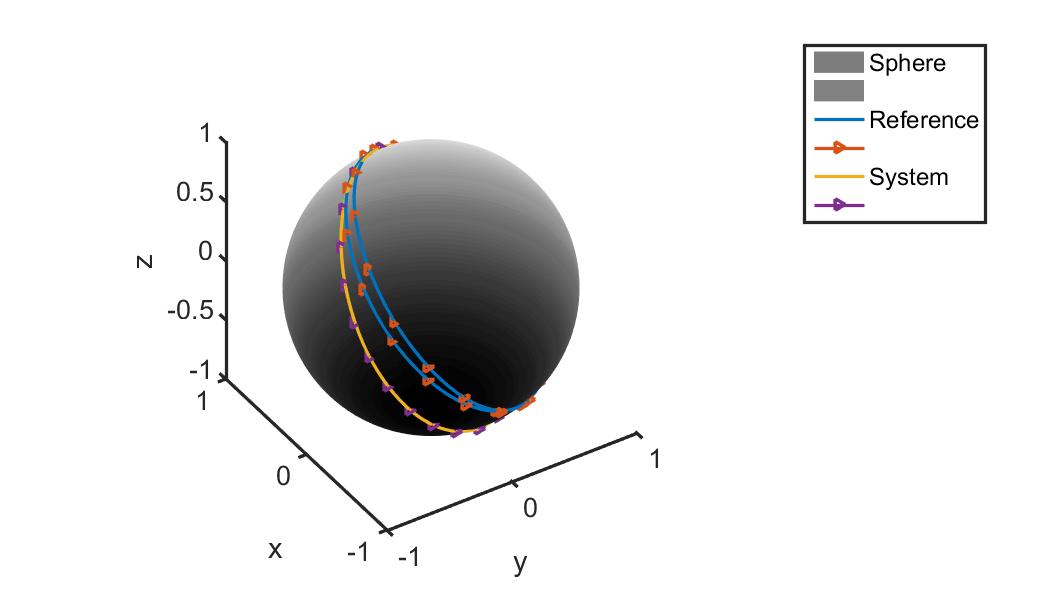}
\end{subfigure}
\begin{subfigure}[b]{0.15\textwidth}
  \centering
  \includegraphics[scale=0.06]{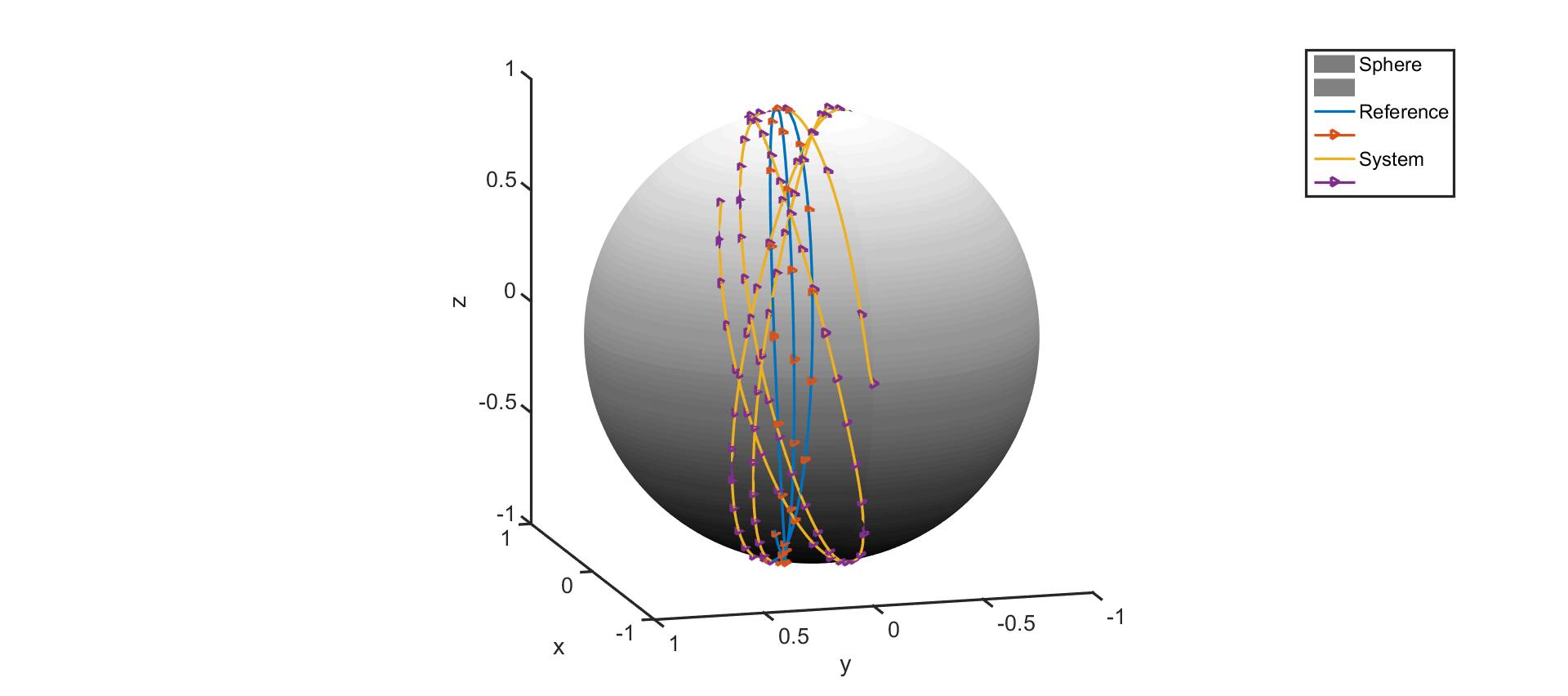}
\end{subfigure}
  \caption{3D visualization of tracking problem for first set of initial conditions}
\end{figure}
\begin{figure}[!ht]
  \centering
  \begin{subfigure}[b]{0.15\textwidth}
  \includegraphics[scale=0.1]  {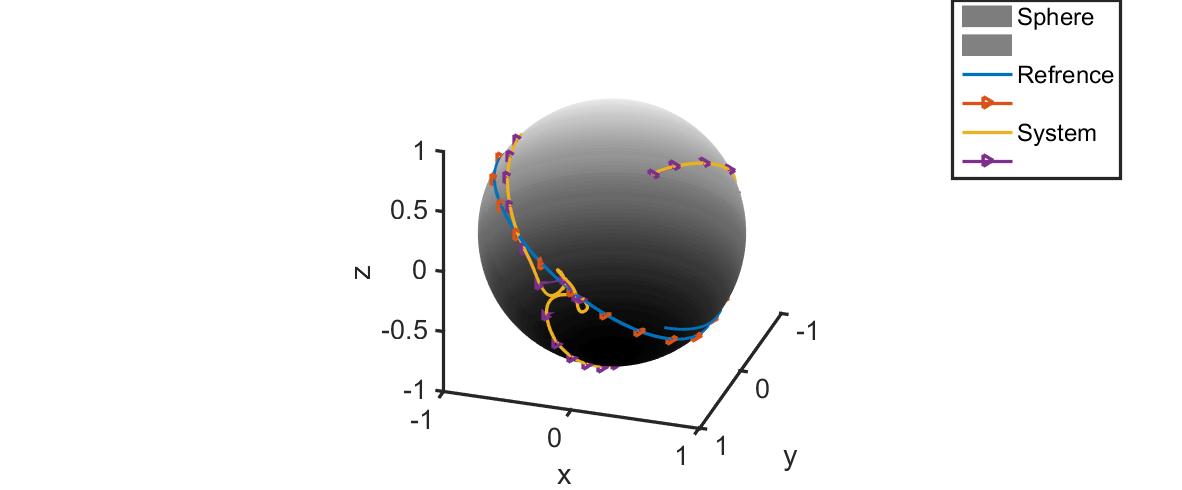}
  \end{subfigure}
\begin{subfigure}[b]{0.15\textwidth}
  \centering
  \includegraphics[scale=0.1]{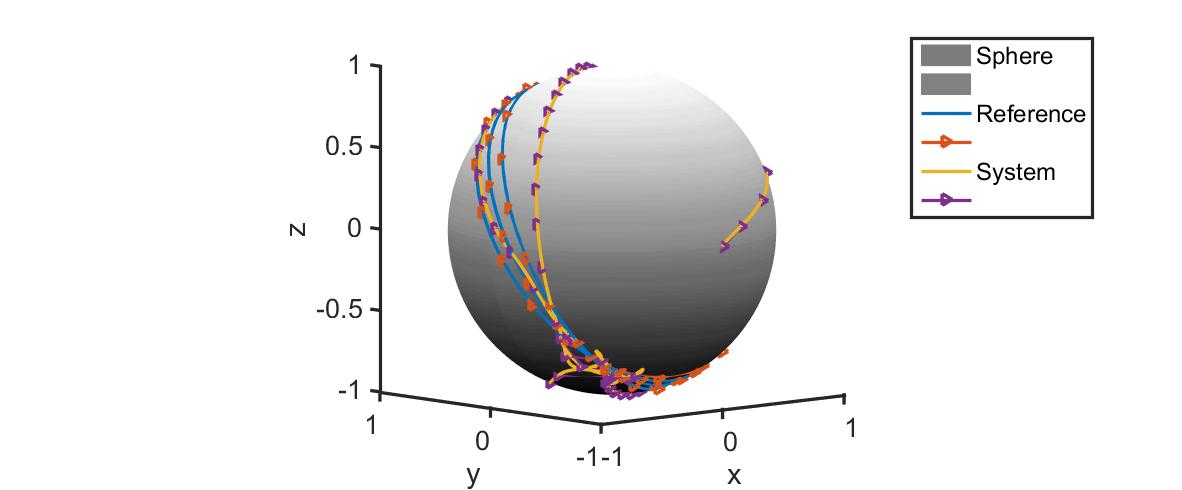}
\end{subfigure}
\begin{subfigure}[b]{0.15\textwidth}
  \centering
  \includegraphics[scale=0.1]{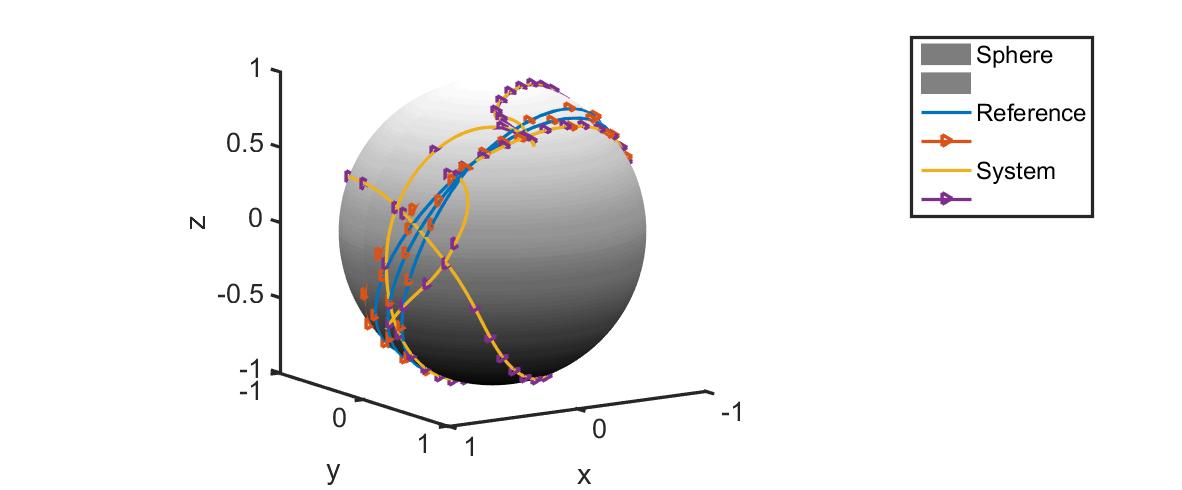}
\end{subfigure}
  \caption{3D visualization of tracking problem for second set of initial conditions}
\end{figure}
The constrained affine connection on $S^2$ is given by
\[ \stackrel{\mathcal{D}}{\nabla} _{\dot{\boldsymbol{\gamma}}} \dot{\boldsymbol{\gamma}} =  \ddot{\boldsymbol{\gamma}}(t) + || \dot{\boldsymbol{\gamma}}||_2^2\boldsymbol{\gamma}.
\]
Therefore, from \eqref{113}, the system trajectory $\boldsymbol{\gamma}(t)$ for any spherical pendulum satisfies the following equation
\begin{equation}\label{130}
  \ddot{\boldsymbol{\gamma}}(t) + || \dot{\boldsymbol{\gamma}}||_2^2\boldsymbol{\gamma} = P_{\mathcal{D}_{\boldsymbol{\gamma}}}(u).
\end{equation}
The reference trajectory is generated by a dummy spherical pendulum with the following initial conditions $(\boldsymbol{\gamma}_{ref}(0),\dot{\boldsymbol{\gamma}}_{ref}(0))=\begin{pmatrix}
                                          \frac{1}{\sqrt{2}} & 0 & \frac{1}{\sqrt{2}} & 3 & 0 & -3
                                        \end{pmatrix}^T$ and $u = \begin{pmatrix}
                                          1 & 2 & -1
                                        \end{pmatrix}^T$.
The initial conditions for the system trajectory is given by $
  (\boldsymbol{\gamma}(0),\dot{\boldsymbol{\gamma}}(0))=\begin{pmatrix}
                                          -1 & 0 & 0& 0 & 1 & 0
                                        \end{pmatrix}^T
$. Theorem 1 is applied to compute the tracking control given in \eqref{130} with $F_{diss}=-4$, $K_p = 3.7$. The system trajectory is generated using ODE45 solver of MATLAB. The reference (in blue) and system trajectory (in red) are compared in all $3$ coordinates in figures ~\ref{fig1}, ~\ref{fig2} and ~\ref{fig3}. We consider another set of initial conditions as follows.
$(\boldsymbol{\gamma}_{ref}(0),\dot{\boldsymbol{\gamma}}_{ref}(0))=\begin{pmatrix}
                                          \frac{1}{\sqrt(3)} & \frac{1}{\sqrt(3)} & \frac{1}{\sqrt(3)}& 3 & 0 & -3
                                        \end{pmatrix}^T$ and $u = \begin{pmatrix}
                                          1 & 2 & 1
                                        \end{pmatrix}^T$ for the dummy spherical pendulum. The initial conditions for the system trajectory are $
  (\boldsymbol{\gamma}(0),\dot{\boldsymbol{\gamma}}(0))=\begin{pmatrix}
                                          0 & -1 & 0& 1 & 2 & 2
                                        \end{pmatrix}^T.$ Theorem 1 is applied to compute the tracking control with $F_{diss} =-5.7$ and $K_p= 4$. The reference (in blue) and system trajectory (in red) are compared in all $3$ coordinates in figures ~\ref{fig4}, ~\ref{fig5} and ~\ref{fig6}.

\subsection{AGAT on Lissajous curve}
\begin{figure}
  \centering
  \includegraphics[scale=0.3]{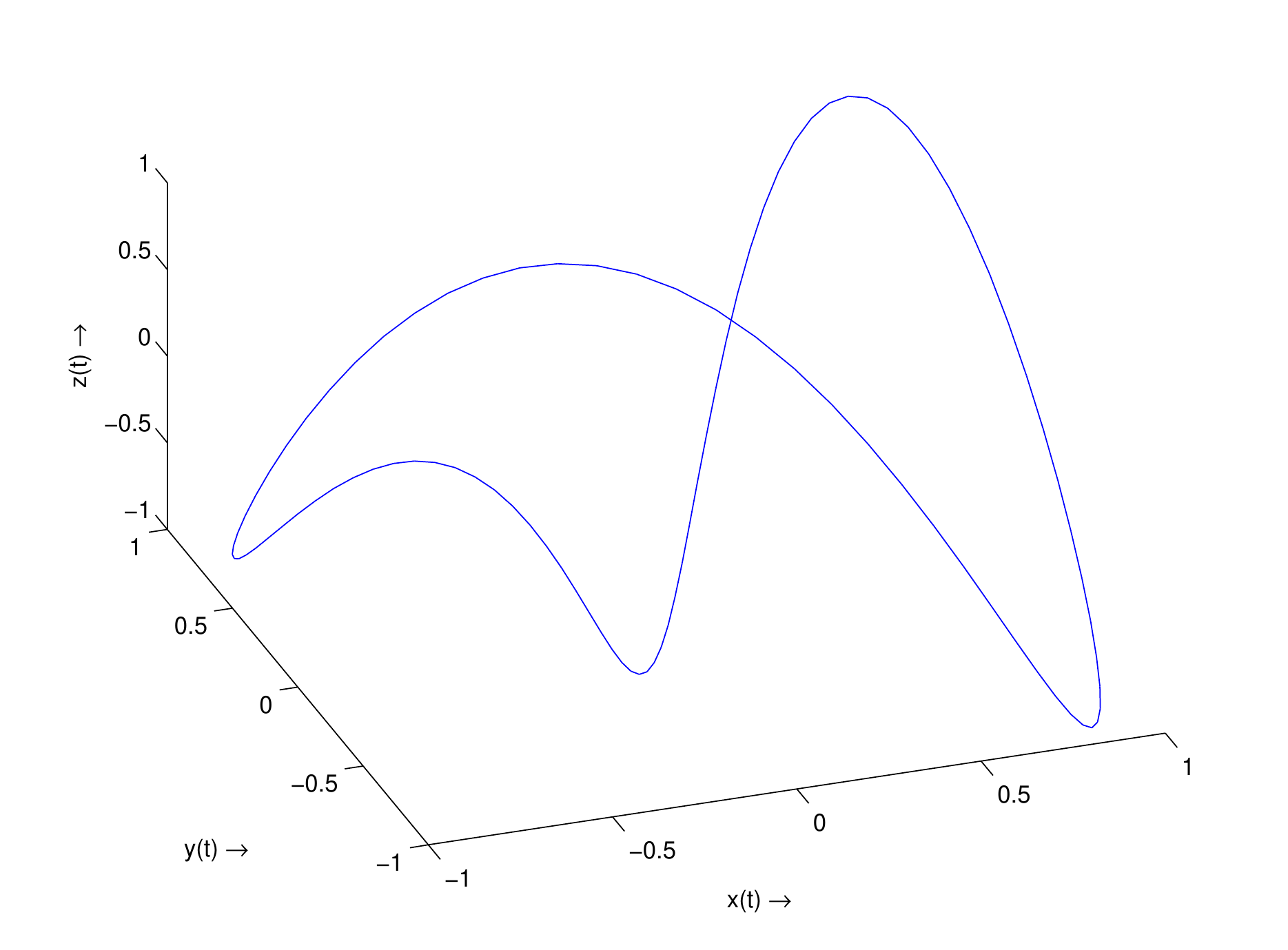}
  \caption{The space curve L}\label{lissfig}
\end{figure}
A Lissajous curve in $3$ dimensions (shown in figure ~\ref{lissfig}) is a $1-$ dimensional smooth, connected, compact manifold in $\mathbb{R}^3$. It is denoted by $L$ and defined as $L = h^{-1}\{\begin{pmatrix}
                                                             0 & 0
                                                           \end{pmatrix}^T\}$
                                                           where $h : \mathbb{R}^3 \to \mathbb{R}^2$ is given by $h(x,y,z) = \begin{pmatrix}
              x^2+y^2+z^2-2xyz -1 & 4x^2y -2xz-y
            \end{pmatrix}^T$.Therefore $T_\textbf{x} L= \{v \in \mathbb{R}^3: \mathrm{D}h(\textbf{x}) v=0\}$ for $\textbf{x} \in L$.

\subsection{Navigation function:} We consider $\boldsymbol{\psi}: L \to \mathbb{R}$ given as $\boldsymbol{\psi}(x,y,z) = x$. It is observed that $\boldsymbol{\psi}$ has a unique minimum at $\begin{pmatrix}
-1 & 0 & 0
\end{pmatrix}^T$ and a unique maximum at $\begin{pmatrix}
1 & 0 & 0
\end{pmatrix}^T$. Using parameterizations $\psi_1(t)= \begin{pmatrix}
                             \cos(t) & \sin(2t) & \sin(3t)
                           \end{pmatrix}^T, t \in (-\pi,\pi)$ around $\begin{pmatrix} -1 & 0& 0 \end{pmatrix}^T$ and the parameterization $\psi_2(t)= \begin{pmatrix}
                             \cos(t) & \sin(2t) & -\sin(3t)
                           \end{pmatrix}^T, t \in (-\pi,\pi)$ around $\begin{pmatrix} 1 &0 &0 \end{pmatrix}^T$, it is verified that $Hess ({\psi_i(t)})|_{t=0} \neq 0$, $i=1,2$. Therefore $\psi$ is a navigation function.
\subsubsection{Configuration error map} The configuration error map $\textbf{E}:L \times L \to L$ is chosen as
\begin{equation}\label{confige}
  \textbf{E}(q_1,q_2)= \begin{pmatrix}
                -\frac{\langle q_1,q_2\rangle}{|q_1||q_2|}_{\mathbb{R}^3} \\
-2\frac{|q_1 \times q_2|{\langle q_1,q_2\rangle}_{\mathbb{R}^3}}{|q_1|^2|q_2|^2}  \\  4\frac{|q_1 \times q_2|{{\langle q_1,q_2\rangle}_{\mathbb{R}^3}}^2}{|q_1|^3|q_2|^3 }-\frac{|q_1 \times q_2|}{|q_1||q_2| }
              \end{pmatrix}.
\end{equation}
It is observed that $\textbf{E}(q_1,q_2) \in L$ and that $(\boldsymbol{\psi}\circ \textbf{E})(q_1,q_2)$ is symmetric. As $\textbf{E}(q,q)= \begin{pmatrix}
                                    -1 & 0 & 0
                                  \end{pmatrix}$ hence, $\textbf{E}(q,q)$ is the minimum of $\boldsymbol{\psi}$. Therefore, $(\boldsymbol{\psi}, \textbf{E})$ is a compatible pair according to Definition 1. We define $\beta = \frac{|q_1 \times q_2|}{|q_1||q_2|}$ and $\textbf{E} = \begin{pmatrix}
 E_1& E_2&E_3
                                  \end{pmatrix}^T$. Observe that $\beta = \sqrt{1- E_1^2}$. Therefore, $\mathrm{d}_1\textbf{E}(q_1,q_2) =\begin{pmatrix} \frac{\partial E_1 }{\partial q_1} & \frac{\partial E_2 }{\partial q_1}& \frac{\partial E_3 }{\partial q_1} \end{pmatrix}^T$ and, $\frac{\partial \beta }{\partial q_1} = -\frac{E_1}{\beta}\frac{\partial E_1}{ \partial q_1}$, $\frac{\partial \beta }{\partial q_2} = -\frac{E_1}{\beta}\frac{\partial E_1}{ \partial q_2}$.The $(1,1)$ tensors are $3 \times 3$ matrices given as
\begin{align}\label{onet}
\mathrm{d}_1\textbf{E}(q_1,q_2)&= \begin{pmatrix}
                            -E_1\frac{q_1^T}{2|q_1|^2}- \frac{q_2^T}{|q_1||q_2|} \\
   2E_1 \frac{\partial \beta }{\partial q_1}+2\beta \frac{\partial E_1 }{\partial q_1} \\
   (4E_1^2 - 1)\frac{\partial \beta }{\partial q_1} + 8 \beta E_1 \frac{\partial E_1 }{\partial q_1} \end{pmatrix} \quad \text{and},\\ \nonumber
\mathrm{d}_2\textbf{E}(q_1,q_2)&=\begin{pmatrix}
 -E_1\frac{q_2^T}{2|q_2|^2}- \frac{q_1^T}{|q_1||q_2|} \\
   2E_1 \frac{\partial \beta }{\partial q_2}+2\beta \frac{\partial E_1 }{\partial q_2} \\
   (4E_1^2 - 1)\frac{\partial \beta }{\partial q_2} + 8 \beta E_1 \frac{\partial E_1 }{\partial q_2}
   \end{pmatrix}.
\end{align}
The $(2,1)$ tensors are $3 \times 3 \times 3$ arrays given as
\begin{align}\label{twot}
&\mathrm{d}_1(\mathrm{d}_1\textbf{E})(q_1,q_2)=\\ \nonumber
&\begin{pmatrix}
                     \frac{q_1 q_2^T}{2|q_1|^3|q_2|}- I_3\frac{E_1}{2|q_1|^2} + E_1\frac{q_1 q_1^T}{2|q_1|^4} -\{\frac{\partial E_1 }{\partial q_1}\}^T\frac{q_1^T}{2|q_1|^2} \\ \nonumber
4 sym\big (\{\frac{\partial E_1}{\partial q_1}\}^T \frac{\partial \beta }{\partial q_1} \big)+2E_1 \frac{\partial^2 \beta }{\partial q_1^2}+2\beta \frac{\partial^2 E_1 }{\partial q_1^2}\\  \nonumber
  16E_1 sym\big ({\frac{\partial E_1}{\partial q_1}}^T \frac{\partial \beta }{\partial q_1} \big)+\alpha \frac{\partial^2 \beta }{\partial q_1^2}+8\beta E_1 \frac{\partial^2 E_1 }{\partial q_1^2}+ 8 \beta {\frac{\partial E_1}{\partial q_1}}^T \frac{\partial E_1}{\partial q_1}
\end{pmatrix}\\ \nonumber
&\mathrm{d}_2(\mathrm{d}_1\textbf{E})(q_1,q_2)=\\ \nonumber
&\begin{pmatrix}
-\frac{I_3}{|q_1||q_2|} + \frac{q_2 q_2^T}{2 |q_2|^3|q_1|} - \frac{q_1^T}{2|q_1|^2}\frac{\partial E_1 }{\partial q_2}\\
2{\frac{\partial E_1}{\partial q_2}}^T \frac{\partial \beta }{\partial q_1}
+2{\frac{\partial \beta}{\partial q_2}}^T \frac{\partial E_1 }{\partial q_1}
+2E_1 \frac{\partial^2 \beta }{\partial q_2 \partial q_1}+2\beta \frac{\partial^2 E_1 }{\partial q_2 \partial q_1}\\
8E_1\big({\frac{\partial E_1}{\partial q_2}}^T \frac{\partial \beta }{\partial q_1}
+{\frac{\partial \beta}{\partial q_2}}^T \frac{\partial E_1 }{\partial q_1}\big)
+ \alpha \frac{\partial^2 \beta }{\partial q_2\partial q_1}+8\beta\big( E_1\frac{\partial^2 E_1 }{\partial q_2\partial q_1}+  {\frac{\partial E_1}{\partial q_2}}^T \frac{\partial E_1}{\partial q_1}\big)
\end{pmatrix} \nonumber
\end{align}
The tensors $\mathrm{d}_2(\mathrm{d}_2\textbf{E})$ and $\mathrm{d}_1(\mathrm{d}_2\textbf{E})$ are obtained similarly as $\textbf{E}(q_1,q_2)$ is symmetric in $q_1$ and $q_2$.

\subsubsection{AGAT results}
We consider a particle moving on the curve $L$. The equations of motion of the particle are given by the geodesic $\boldsymbol{\gamma}(t)$ on $L$ for $t \in \mathbb{R}^+$. Therefore,
\begin{equation}\label{geod1}
\ddot{\boldsymbol{\gamma}}(t)= \lambda_1 \mathrm{d}h(1) + \lambda_2 \mathrm{d}h(2)
\end{equation}
as $\ddot{\boldsymbol{\gamma}}(t) \in \mathcal{D}^\perp_\gamma$ for all $t$. Since $\gamma(t) \in L$ implies $h \circ \boldsymbol{\gamma}(t)=0$, therefore,
\begin{equation}\label{geod2}
\dot{\boldsymbol{\gamma}}^T \mathrm{D}^2h \dot{\boldsymbol{\gamma}} + \mathrm{D}h \ddot{\boldsymbol{\gamma}}=0.
\end{equation}
From \eqref{geod1} and \eqref{geod2} we obtain $\lambda_1$, $\lambda_2$ and hence the geodesic curve $\boldsymbol{\gamma}$. Therefore the affine connection on $L$ is given as
\begin{equation}\label{geodL}
\stackrel{\mathcal{D}}{\nabla}_{\dot{\boldsymbol{\gamma}}}\dot{\boldsymbol{\gamma}} =\ddot{\boldsymbol{\gamma}}(t)- \lambda_1 \mathrm{d}h(1) - \lambda_2  \mathrm{d}h(2)
\end{equation}
We consider the reference trajectory $\boldsymbol{\gamma_{ref}}(t)=\begin{pmatrix} \cos(\sin(t))& \sin(2\sin(t)) &\sin(3\sin(t)) \end{pmatrix}^T, t \geq 0$ with $(\boldsymbol{\gamma_{ref}}(0),\dot{\boldsymbol{\gamma}}_{ref}(0))= \begin{pmatrix}
1& 0 &0 & 0 & 2 & 3
\end{pmatrix}^T$. The initial conditions for system trajectory are $ (\boldsymbol{\gamma}(0),\dot{\boldsymbol{\gamma}}(0))= \begin{pmatrix}
-0.82& 0.9386 & 0.9672 & -1.197 & 1.1346 & 2.0798
\end{pmatrix}^T$.  Theorem 1 is applied to compute the tracking control in \eqref{thm11} with $F_{diss}=-1.2$, $K_p = 5.4$. The system trajectory is generated using ODE45 solver of MATLAB. The reference (in blue) and system trajectory (in red) are compared in all $3$ coordinates in figures ~\ref{liss11}, ~\ref{liss12}, ~\ref{liss13}. Another simulation is performed with the reference trajectory $\boldsymbol{\gamma_{ref}}(t)=\begin{pmatrix} \cos(t)& \sin(2t) &\sin(3t) \end{pmatrix}^T, t \geq 0$ with $(\boldsymbol{\gamma_{ref}}(0),\dot{\boldsymbol{\gamma}}_{ref}(0))= \begin{pmatrix}
1& 0 &0 & 0 & 2 & 3
\end{pmatrix}^T$. The initial conditions for system trajectory are $ (\boldsymbol{\gamma}(0),\dot{\boldsymbol{\gamma}}(0))= \begin{pmatrix}
0.54 & 0.9093 & 0.1411 & -0.8415 & -0.8323 & -2.97
\end{pmatrix}^T$. Theorem 1 is applied to compute the tracking control in \eqref{thm11} with $F_{diss}=-1.6$, $K_p = 5.3$.The reference (in blue) and system trajectory (in red) are compared in all $3$ coordinates in figures ~\ref{liss21}, ~\ref{liss22}, ~\ref{liss23}.
\begin{figure}[h!]
  \centering
  \begin{subfigure}[b]{0.15\textwidth}
  \includegraphics[scale=0.1]  {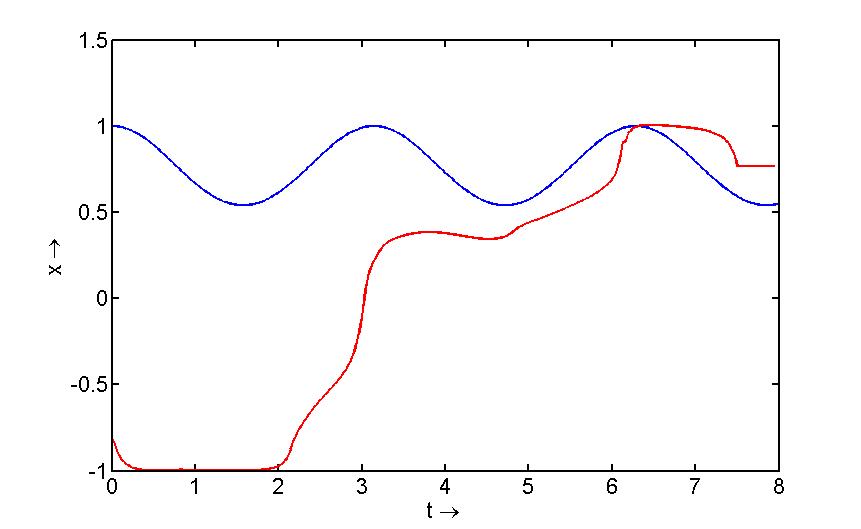}
  \caption{x coordinate}\label{liss11}
  \end{subfigure}
\begin{subfigure}[b]{0.15\textwidth}
  \centering
  \includegraphics[scale=0.1]{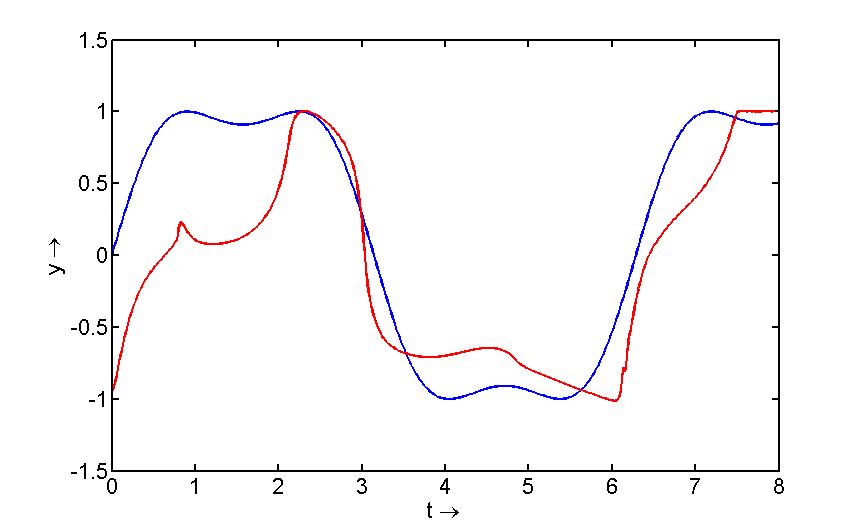}
  \caption{y coordinate}\label{liss12}
\end{subfigure}
\begin{subfigure}[b]{0.15\textwidth}
  \centering
  \includegraphics[scale=0.1]{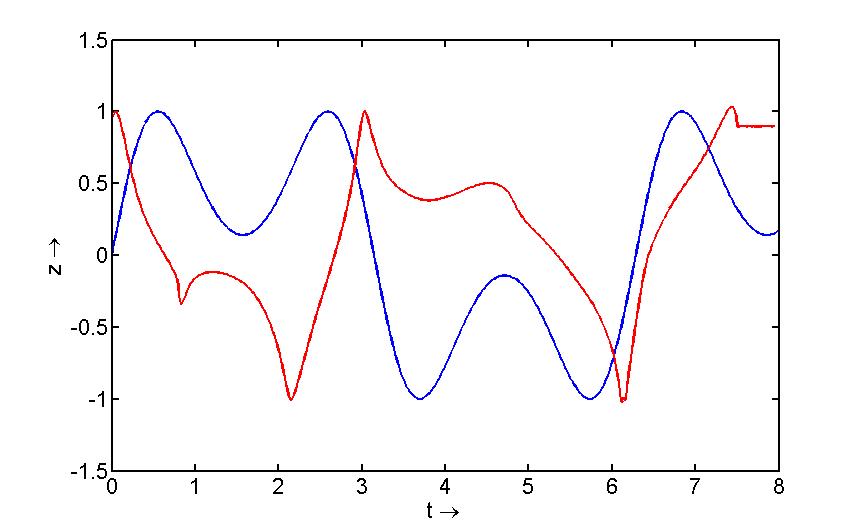}
  \caption{z coordinate}\label{liss13}
  \end{subfigure}
  \caption{Tracking results for first set of initial conditions}\label{lisseg1}
\end{figure}

\begin{figure}[h!]
  \centering
  \begin{subfigure}[b]{0.15\textwidth}
  \includegraphics[scale=0.1]  {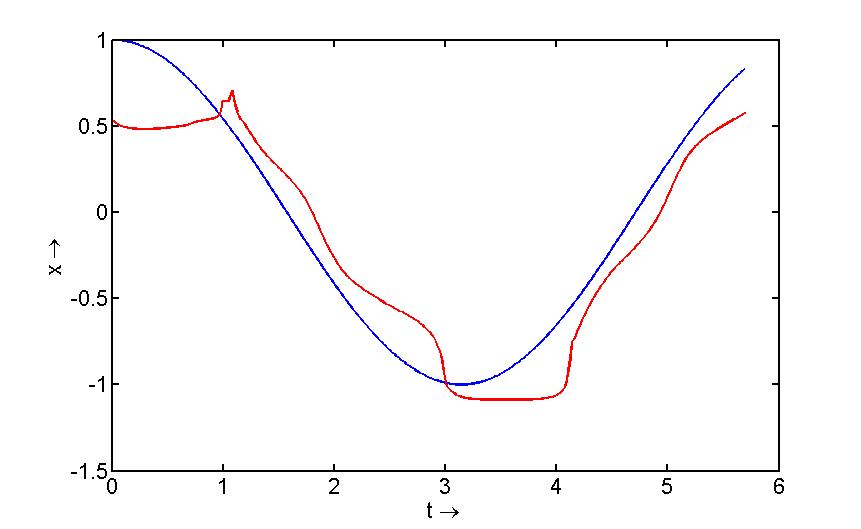}
  \caption{x coordinate}\label{liss21}
  \end{subfigure}
\begin{subfigure}[b]{0.15\textwidth}
  \centering
  \includegraphics[scale=0.1]{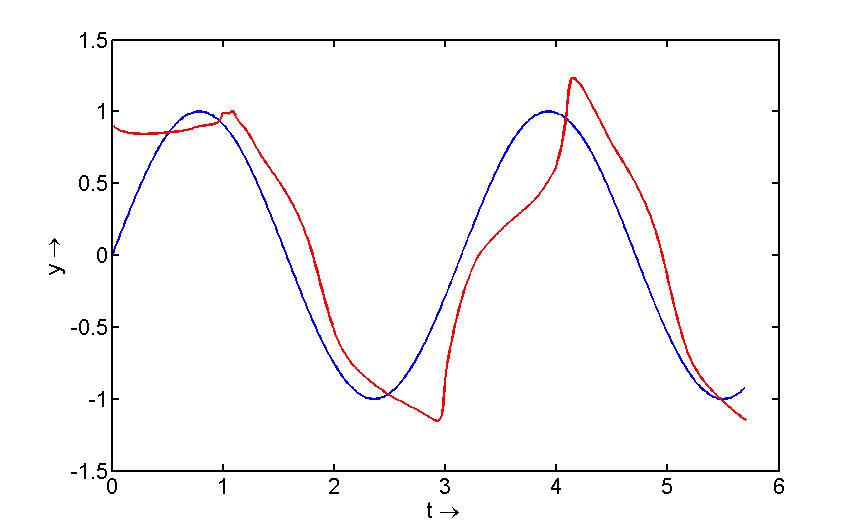}
  \caption{y coordinate}\label{liss22}
\end{subfigure}
\begin{subfigure}[b]{0.15\textwidth}
  \centering
  \includegraphics[scale=0.1]{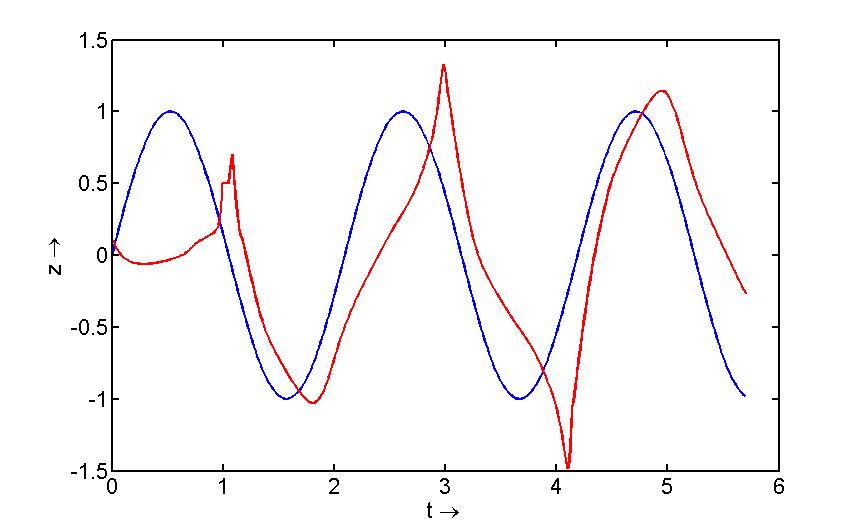}
  \caption{z coordinate}\label{liss23}
  \end{subfigure}
  \caption{Tracking results for second set of initial conditions}\label{lisseg2}
\end{figure}

\bibliographystyle{plain}
\bibliography{aps1}

\end{document}